\documentclass[11pt]{article}

\usepackage{comment}
\usepackage{graphicx}
\usepackage{amsthm}
\usepackage{amsmath}
\usepackage{amssymb}
\usepackage{soul}
\usepackage{mdframed, framed, color}
\usepackage[linesnumbered,algonl,boxed, ruled]{algorithm2e}
\usepackage{lineno}

\newcommand{\OO}{{\cal O}}
\newcommand{\LL}{\mathcal{L}}
\newcommand{\RR}{\mathcal{R}}
\newcommand{\CC}{\mathcal{C}}
\newcommand{\BB}{\mathcal{B}}
\newcommand{\FF}{\mathcal{F}}
\newcommand{\MM}{\mathcal{M}}
\newcommand{\GG}{\mathcal{G}}

\newcommand{\HH}{\mathcal{H}}
\newcommand{\XX}{\mathcal{X}}
\newcommand{\ZZ}{\mathcal{Z}}

\newcommand{\JJ}{\mathcal{J}}
\newcommand{\EE}{\mathcal{E}}

\newcommand{\QQ}{\mathcal{Q}}
\newcommand{\WW}{\mathcal{W}}
\newcommand{\YY}{\mathcal{Y}}

\newcommand{\DD}{\mathcal{D}}

\newtheorem{definition}{Definition}
\newtheorem{property}{Property}
\newtheorem{lemma}{Lemma}
\newtheorem{theorem}{Theorem}

\newtheorem{observation}{Observation}

\newtheorem{conjecture}
{Conjecture}

\title{Layouts for Plane Graphs on Constant Number of Tracks}

\date{}
\author{Jiun-Jie Wang\\
Email: jiunjiew@buffalo.edu}
\begin{document}
\maketitle

\abstract{
A \emph{$k$-track} layout of a graph consists of a vertex $k$ colouring, and a total order of each vertex colour class, such that between each pair of colour classes no two edges cross.
A \emph{$k$-queue} layout of a graph consists of a total order of the vertices, and a partition of the edges into $k$ sets such that no two edges that are in the same set are nested with respect to the vertex ordering.
The \emph{track number} (\emph{queue number}) of a graph $G$, is the minimum $k$ such that $G$ has a $k$-track ($k$-queue) layout.

This paper proves that every $n$-vertex plane graph has constant-bound track and queue numbers.
The result implies that every plane has a 3D crossing-free straight-line grid drawing
in $O(n)$ volume.
The proof utilizes a novel graph partition technique.
}

\section{Introduction}

A \emph{track} layout of a graph consists of a vertex $k$ colouring, and a total order of each vertex colour class, such that between each pair of colour classes no two edges cross.
A \emph{queue} layout of a graph consists of a total order of the vertices, and a partition of the edges into sets (called queues) such that no two edges that are in the same set are nested with respect to the vertex ordering. The minimum number of queues in a queue layout of a graph is its \emph{queue number}.
Track layouts have been extensively studied in \cite{D15, DFJW12, DMW05, DMW13, DPW04, DW04, DW05, FLW02}.
Queue layouts have been introduced by Heath, Leighton, and Rosenberg \cite{HLR92, HR92} and have been extensively studied from \cite{BFP10, DMW05, DMW13, DPW04, DW05, EI71, H04, HLR92, HR92, P92, RM95, SS00, T72, W05, W08}. Both track and queue layouts have applications in parallel process scheduling, fault-tolerant processing, matrix computations, and sorting networks (see \cite{P92} for a survey).
Queue layouts of directed acyclic graphs \cite{BCLR96, HP99, HPT99, P92} and posets \cite{HP97, P92} have also been investigated.

The question in Heath et al. \cite{HLR92, HR92}, whether the queue number of a planar graphs is constant-bound (it also leads to constant-bound track number), remains open. Heath et al. \cite{HLR92, HR92} conjectured that the question has an affirmative answer.
However, Pemmaraju \cite{P92} conjectured that every planar graph has $O(\log n)$ queue number. Also,
he conjectured that this is the correct lower bound.
Up to now, the best known lower bound is still constant-bound.
On the other hand, the well-known upper bound for the queue number of planar graphs had remained stagnant as $O(\sqrt{n})$ roughly two decades. This upper bound utilizes the fact that planar graphs have path width at most $O(\sqrt{n})$.
Recently, the upper bounds of queue and track numbers for planar graphs were reduced to $O(\log^2 n)$, by Di Battista, Frati and Pach \cite{BFP10} and $O(\log n)$, by Vida Djumovic \cite{D15}, respectively.

In this paper, we provide a layout on constant number of tracks for a plane graph.
Our result attempts to break Pemmaraju's conjecture in a positive direction.
The proof that a plane graph has constant-bound track number is simple.
It utilizes a novel graph partition technique.
In particular, our main result states that every $n$-vertex plane graph has such a graph partition and it leads to $O(1)$-track layouts for plane graphs.

One of the most important motivations for studying queue layouts is 3D crossing-free straight-line grid drawing in a small volume.
Particularly, a 3D crossing-free straight-line grid drawing of a graph is a placement of the vertices at distinct points in a 3D grid, and the straight-line representing the edges are pairwise non-crossing.
One of the most important open problems that Felsner et al. \cite{FLW02} present in graph drawing questions  is whether planar graphs have 3D crossing-free straight-line grid drawings in a linear volume.
A 3D crossing-free straight-line grid drawing with volume $X \times Y \times Z$ is an $X \times Y \times Z$ drawing that fits in an axis-aligned box with side lengths $X-1$, $Y-1$, and $Z-1$.
The following theorem has been established in \cite{DMW05, DPW04}.
\begin{theorem}
An $n$-vertex graph $G$ has a 3D crossing-free straight grid drawing in an $O(1) \times O(1) \times O(n)$ volume, if and only if $G$ has a constant-bound queue number. (constant-bound track number.)
\end{theorem}

The road map for this paper is as follows:
the first half part from Sections \ref{sec:prelim} to \ref{sec:cons-tracks}
explain the basic framework and ideas for this article.
The second half part explain more details in the first half part of this article.
\section{Preliminaries}\label{sec:prelim}
In this section, Some definitions and important preliminary results are given.
Definitions not mentioned here are standard. A graph $G=(V,E)$ is called
{\em planar} if it can be drawn on the plane with no edge
crossings. Such a drawing is called a {\em plane embedding} of $G$.
A {\em plane graph} is a planar graph with a fixed plane embedding.

A \emph{layerlike} graph $\Pi$ is a graph whose vertices are partitioned and placed on contiguous layers such that
no edge is placed between any two non-contiguous layers and no edges are crossing.
Given a layerlike graph $\Pi$, a \emph{down-pointing} triangle $\triangledown$ is a a cycle $(l, \cdots, r, m)$
that vertices on the cycle $(l, \cdots, r, m)$
are on the two contiguous layers
where the path from $l$ to $r$ are on the upper layer and the vertex $m$ is on the lower layer.
A \emph{bowl} $\heartsuit$ is a cycle $(l, \cdots, r)$
that the cycle are on the same layer
where each vertex of the cycle is on the same layer. In Fig. \ref{fig:perfect-layer-graph},
vertices $(b_8, b_9, b_{10}, b_{11}, b_{12}, b_{13})$ form a bowl in the composite-layerlike graph $\GG$.

\begin{definition}
A \emph{composite-layerlike} graph $\GG$ can be recursively defined as follows:
$\GG$ consists of a layerlike graph $\Pi$ such that
each bowl $\heartsuit$ of $\GG$ has a smaller composite-layerlike graph $\GG_1$
where $\GG_1$'s first layer is the bowl $\heartsuit$,
and each down-pointing triangle $\triangledown$ has a composite-layerlike graph $\GG_2$
where the first layer of $\GG_2$ is the upper layer of $\triangledown$.
\end{definition}

An edge $e=(u, v)$ is called a \emph{chord} if
both end-vertices $u$ and $v$
are on the same layer in a composite-layerlike graph $\GG$.
A \emph{region} $\WW$, rooted at a vertex $r$ in a composite-layerlike graph $\GG$, consists of a left boundary $\BB^L$ and a right boundary $\BB^R$ such that $\WW$ satisfies (1): $B^L$ and $\BB^R$ are two paths walking along contiguous layers from the vertex $r$ to two different vertices on lower layers in $\GG$, and
(2) $\WW$ is a separator of the composite-layerlike graph $\GG$. Also, we denote the left and right boundaries of a region $\WW$ as $\BB^L(\WW)$ and $\BB^R(\WW)$, respectively.
Consider a region $\WW$ rooted at a vertex $r$ in a composite-layerlike graph $\GG$.
A composite-layerlike graph $\GG(\WW)$ is \emph{induced} by $\WW$ if
$\GG(\WW)$ is a subgraph of $\GG$ inside by the two boundaries $\BB^L(\WW)$ and $\BB^R(\WW)$.
Also,
we denote $\WW^{\MM}$ as the maximum region bounded by the leftmost and rightmost boundaries of $\GG$.
Obviously, a composite-layerlike graph $\GG$ is a maximum composite-layerlike graph $\GG(\WW^{\MM})$
induced by the maximum region $\WW^{\MM}$.

A \emph{ladder} $\HH$ is defined to
consist of contiguous tracks.
A \emph{layout} of a composite-layerlike graph $\GG$ in a ladder $\HH$ is defined to be an arbitrary vertices's partition of $\GG$ on tracks of $\HH$.
For a layout of a composite-layerlike graph $\GG$ in a ladder $\HH$,
a set of chords $\{e_1=(u_1, v_1), e_2=(u_2, v_2), \cdots, e_q=(u_q, v_q)\}$ are called \emph{nest} if $\{e_1, e_2, \cdots, e_q\}$ are placed on a track in $\HH$ as the order:
$(u_1,$ $u_2,$ $\cdots,$ $u_q,$ $v_q,$ $\cdots,$ $v_2,$ $v_1)$.
A set of edges $\{e_1=(u_1, v_1),$ $e_2=(u_2, v_2),$ $\cdots,$ $e_q=(u_q, v_q)\}$ are called \emph{$X$-cross} if
$(u_1,$ $u_2,$ $\cdots,$ $u_q)$ are orderly placed  as $(u_1,$ $u_2,$ $\cdots,$ $u_q)$　on a track and
$(v_1,$ $v_2,$ $\cdots,$ $v_q)$ are reversely placed as
$(v_q,$ $\cdots,$ $v_2,$ $v_1)$ on another track in $\HH$.

Given an edge $e=(u, v)$ a layout in $\HH$, let $\LL_{\HH}(u)$ and $\LL_{\HH}(v)$ be the track numbers where the vertices $u$ and $v$ placed in $\HH(\GG)$, respectively. The \emph{gap} of an edge $e=(u, v)$ is the absolute difference $|\LL_{\HH}(u)-\LL_{\HH}(v)|$.
Also, the \emph{queue number} on a track is defined as the maximum size of edges nest on the track, and
the $X$-\emph{crossing number} for any two tracks in $\HH$ is defined as the maximum size of edges $X$-cross between the two track.
The \emph{distance-number} of a layout in a ladder $\HH$ is defined as the maximum gaps among all edges.

\begin{figure}[t]
\begin{center}
\includegraphics[width=1\textwidth, angle =0]{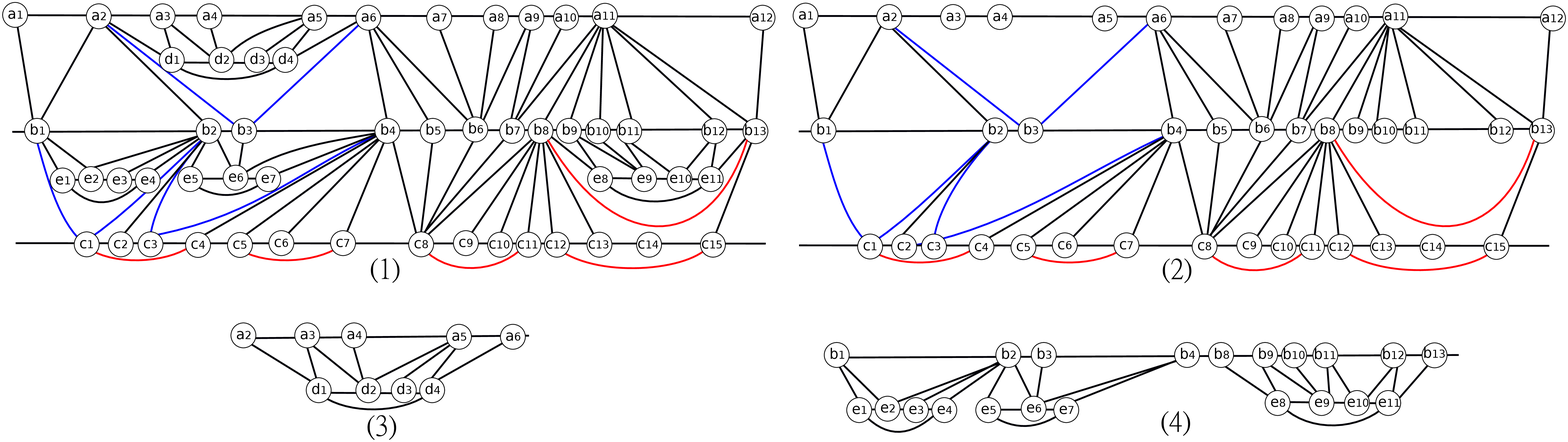}
  \centering
\caption{(1): $\GG$ is a composite-layerlike graph;
(2) is a layerlike graph $\Pi$ of $\GG$ with three layers $(\LL_1, \LL_2, \LL_3)$ and $\Pi$ has three down-pointing triangles and one bowl
$\triangledown_1=(a_2, a_3, a_4, a_5, a_6, b_3)$,
$\triangledown_2=(b_1, b_2, c_1)$,
$\triangledown_3=(b_2, b_3, b_4, c_3)$ and
$\heartsuit_1=(b_8, b_9, b_{10}, b_{11}, b_{12}, b_{13})$;
(3) $\triangledown_1$'s inner vertices can be placed on a layerlike graph $\Pi_1$ which has the same first layer $\LL_1$ with $\Pi$;
(4) the inner vertices of
$\{\triangledown_2, \triangledown_3\}$ and $\heartsuit_1$ can be placed on a  $\Pi_2$ which has the same first layer $\LL_2$ with $\Pi$.
}
\label{fig:perfect-layer-graph}
\end{center}
\vspace{-0.2in}
\end{figure}

\begin{definition}\label{def:well-placed}
A layout of a composite-layerlike graph $\GG$ in
a ladder $\HH$ is called $(\QQ, \XX, \DD)$-\emph{well-placed} if they can be placed in $\HH$ such that
\begin{itemize}
\item each track's queue number is less than $\QQ$,

\item the $X$-crossing number between any two tracks is less than $\XX$,

\item the distance-number in $\HH$ is less than $\DD$, and

\item $\GG$ can be placed as sequential regions in $\HH$;
The sequential regions are denoted as $\tilde{\WW}_{\HH}(\GG)$.
\end{itemize}
\end{definition}

\begin{theorem}\label{thm:wrap}
If a composite-layerlike graph $\GG$ is an $(\QQ, \XX, \DD)$-well-placed layout in $\HH'$, then $\GG$ can be placed as an $(\QQ, \XX, \DD)$-well-placed layout on $2\DD$ tracks in $\HH$.
\end{theorem}

\begin{proof}
Assume that a composite-layerlike graph $\GG$ can
be placed in a ladder $\HH'$
such that
\begin{enumerate}
\item the queue number of each track in $\HH'$ is less than or
equal to $\QQ$,
\item the $X$-crossing number of between any two tracks in $\HH'$ is less than or equal to $\XX$, and
\item the difference $|\LL_{\HH}(u)-\LL_{\HH}(v)|$ is less than or equal to $\DD$ for any edge $e=(u, v)$ in $\HH'$.
\end{enumerate}

Since the total tracks in $\HH'$ could grow beyond constant bound,
we need to wrap $\HH'$ as follows:
move vertices on $(i\times 2\DD + j)$-th track to right of
vertices on $((i-1)\times 2\DD + j)$-th track on the wrapped $\HH$'s $j$-th track.

Now each track $i$ in the wrapped ladder $\HH$,
vertices are from tracks
$(0\times 2\DD + j)$, $(1\times 2\DD + i)$, $(2\times 2\DD + j)$, $\cdots$ in the unwrapped ladder $\HH'$.
And, for each edge $e=(u, v)$ in the wrapped ladder $\HH$,
the edge $e$ comes from pair of tracks
$(0\times 2\DD + \LL_{\HH}(u), 0\times 2\DD + \LL_{\HH}(v))$,
$(1\times 2\DD + \LL_{\HH}(u), 1\times 2\DD + \LL_{\HH}(v))$,
$(2\times 2\DD + \LL_{\HH}(u), 2\times 2\DD + \LL_{\HH}(v))$,
$\cdots$
in the unwrapped ladder $\HH'$.

Because for any edge $e=(u, v)$ in the unwrapped ladder $\HH'$, the difference $|\LL_{\HH}(u)-\LL_{\HH}(v)|$ is at most $\DD$,
only edges on pair tracks
$(0\times 2\DD + \LL_{\HH}(u), 0\times 2\DD + \LL_{\HH}(v))$,
$(1\times 2\DD + \LL_{\HH}(u), 1\times 2\DD + \LL_{\HH}(v))$,
$(2\times 2\DD + \LL_{\HH}(u), 2\times 2\DD + \LL_{\HH}(v))$,
$\cdots$
in the unwrapped ladder $\HH'$
can be placed on the pair tracks $(\LL_{\HH}(u), \LL_{\HH}(v))$ in the wrapped ladder $\HH$.
Also,
for a track $(\LL_{\HH}(u))$ ($\LL_{\HH}(v)$, respectively)
on the wrapped ladder $\HH$, we know that vertices on a track
$i\times 2\DD + \LL_{\HH}(u)$
($i\times 2\DD + \LL_{\HH}(v)$, respectively)
from the unwrapped ladder $\HH'$ are placed at left of
vertices on a track
$(i+1) \times 2\DD + \LL_{\HH}(v)$
($(i+1)\times 2\DD + \LL_{\HH}(u)$, respectively)
from the unwrapped ladder $\HH'$.

Hence
there is no any $X$-crossing edge
between edges from pair tracks
$(i\times 2\DD + \LL_{\HH}(u), i\times 2\DD + \LL_{\HH}(v))$
and
pair tracks
$((i+1)\times 2\DD + \LL_{\HH}(u), (i+1)\times 2\DD + \LL_{\HH}(v))$.

Finally, We can conclude that a composite-layerlike graph $\GG$ can be placed in the wrapped laddder graph $\HH$ such that
\begin{enumerate}
\item the queue number of each track in the wrapped ladder $\HH$ is less than or
equal to $\QQ$,
\item the $X$-crossing number of between any two track in the wrapped ladder $\HH$ is less than or equal to $\XX$, and
\item the number of tracks in the wrapped ladder $\HH$ is $2\DD$.
\end{enumerate}

\end{proof}

\section{A Framework to Construct an $(\QQ, \XX, \DD)$-Well-Placed Layout on Constant Number of Tracks for a Composite-Layerlike Graph $\GG$}\label{sec:perfect-layer-graph-layout}

\begin{algorithm}[ht]
\caption{A Framework to Place an $(\QQ, \XX, \DD)$-Well-Placed Layout on Constant Number Tracks in a Ladder $\HH$ for a Composite-Layerlike Graph $\GG(\WW^{\MM})$}
\label{alg:framework}

\KwIn{A composite-layerlike graph $\GG(\WW^{\MM})$.}

Place the $\WW^{\MM}$'s root on the first track in $\HH$\;

Place the contiguous layers $(2, 3, \cdots)$ of $\WW^{\MM}$ at right of $\WW^{\MM}$'s root on the contiguous $(\ZZ+2, \ZZ+3, \cdots)$ tracks in $\HH$\;

Add the maximum region $\WW^{\MM}$ into the empty first-in-first-out queue $\tilde{\YY}$\;

\While{$\tilde{\YY}$ is not empty}
{

Let $\WW$ rooted at $r$ be the first region in $\tilde{\YY}$\;



Find sequential skeletons $(\Psi_1, \Psi_2, \cdots)$ inside the region $\WW$\;

Place the sequential skeletons $(\Psi_1, \Psi_2, \cdots)$ orderly at the rightmost part in $\HH$ and
from the $(\LL_{\HH}(r)+2\ZZ)$-th track in $\HH$
where $\LL_{\HH}(r)$ is the track number of the vertex $r$ in $\HH$\;

Add the maximum subsequential regions of $\tilde{\WW}_{\HH}(\tilde{\Psi}(\WW))$ into $\tilde{\YY}$
such that each region of the maximum subsequential regions does not root at $r$\;

Remove the region $\WW$ from $\tilde{\YY}$\;

}
Wrap $\HH$\;
\end{algorithm}

In this section, we provide a framework in Algorithm \ref{alg:framework} to place a composite-layerlike graph $\GG$ as $(\QQ, \XX, \DD)$-well-placed layout in $\HH$ on constant number of tracks.
Before we describe the framework, we introduce  a structure \emph{skeleton}
as follows:
\begin{definition}\label{def:skeleton}
Consider a region $\WW=(\BB^L, \BB^R)$ rooted at a vertex $r$. A subgraph $\Psi$ of a composite-layerlike graph $\GG(\WW)$ is called a \emph{skeleton} of $\WW$ if
$\Psi$ consists of
\begin{enumerate}

\item the $\WW$'s root $r$,

\item sequential regions $\tilde{\WW}(\Psi)$ such that
there are subsequential regions $(\WW^L_1, \WW^L_2, \cdots,$ $\WW^M_1,$ $\cdots,$ $\WW^M_m,$ $\WW^R_1, \WW^R_2, \cdots)$ $\subseteq \tilde{\WW}(\Psi)$
where
(1) for each vertex $u\in \BB^L$, each $u$'s child is inside a region in $(\WW^L_1, \WW^L_2, \cdots)$,
(2) for each vertex $v\in \BB^R$, each $v$'s child is inside a region in $(\WW^R_1, \WW^R_2, \cdots)$, and
(3) each $r$'s child is inside a region in $(\WW^M_1, \WW^M_2, \cdots, \WW^M_m)$.
\end{enumerate}

\end{definition}

\begin{figure}[t]
\includegraphics[width=1\textwidth, angle =0]{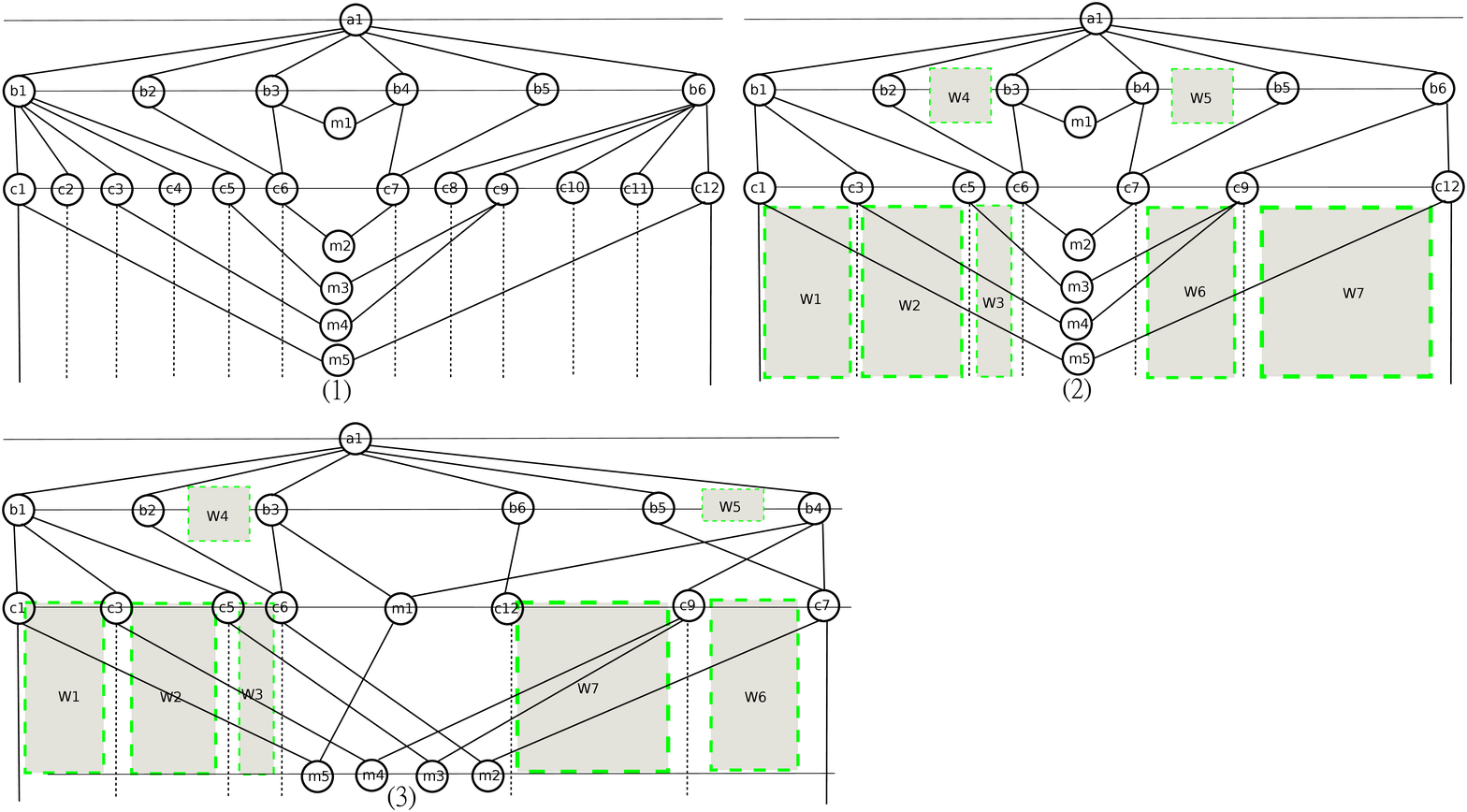}
\caption{(1) shows an example of a region $\WW=(\BB^L, \BB^R)=((a_1, b_1, c_1, \cdots), (a_1, b_6, c_{12}, \cdots))$} and the subgraph induced by $\WW$;
(2) shows a $\WW$'s skeleton that consists of sequential regions
$\WW_1=((b_1, c_1, \cdots), (b_1, c_3, \cdots))$,
$\WW_2=((b_1, c_3, \cdots), (b_1, c_5, \cdots))$,
$\WW_3=((a_1, b_1, c_5, \cdots), (a_1, b_2, c_6, \cdots))$,
$\WW_4=((a_1, b_2, c_6), (a_1, b_2, c_6))$,
$\WW_5=((a_1, b_4, c_7), (a_1, b_5, c_7))$,
$\WW_6=((a_1, b_5, c_7, \cdots), (a_1, b_6, c_9, \cdots))$,
$\WW_7=((b_6, c_9, \cdots), (b_6, c_{12}, \cdots))$.
In addition, the region $\WW_1$ consists of two subregions
$\WW'_1=((b_1, c_1, \cdots), (b_1, c_2, \cdots))$ and
$\WW'_2=((b_1, c_2, \cdots), (b_1, c_3, \cdots))$.
The region $\WW_2$ consists of two subregions
$\WW'_3=((b_1, c_3, \cdots), (b_1, c_4, \cdots))$ and
$\WW'_4=((b_1, c_4, \cdots), (b_1, c_5, \cdots))$.
the region $\WW_7$ consists of three subregions
$\WW'_5=((b_6, c_9, \cdots), (b_6, c_{10}, \cdots))$,
$\WW'_6=((b_6, c_{10}, \cdots), (b_6, c_{11}, \cdots))$ and
$\WW'_7=((b_6, c_{11}, \cdots), (b_6, c_{12}, \cdots))$;
(3): the sequential regions $(\WW_1, \WW_2, \WW_3, \WW_4, \WW_5, \WW_6, \WW_7)$ are placed
as:
$(\WW_1, \WW_2, \WW_3, \WW_4, \WW_7, \WW_6, \WW_5)$
in a ladder $\HH$;
\label{fig:skeleton}
\vspace{-0.2in}
\end{figure}



\begin{figure}[t]
\includegraphics[width=1\textwidth, angle =0]{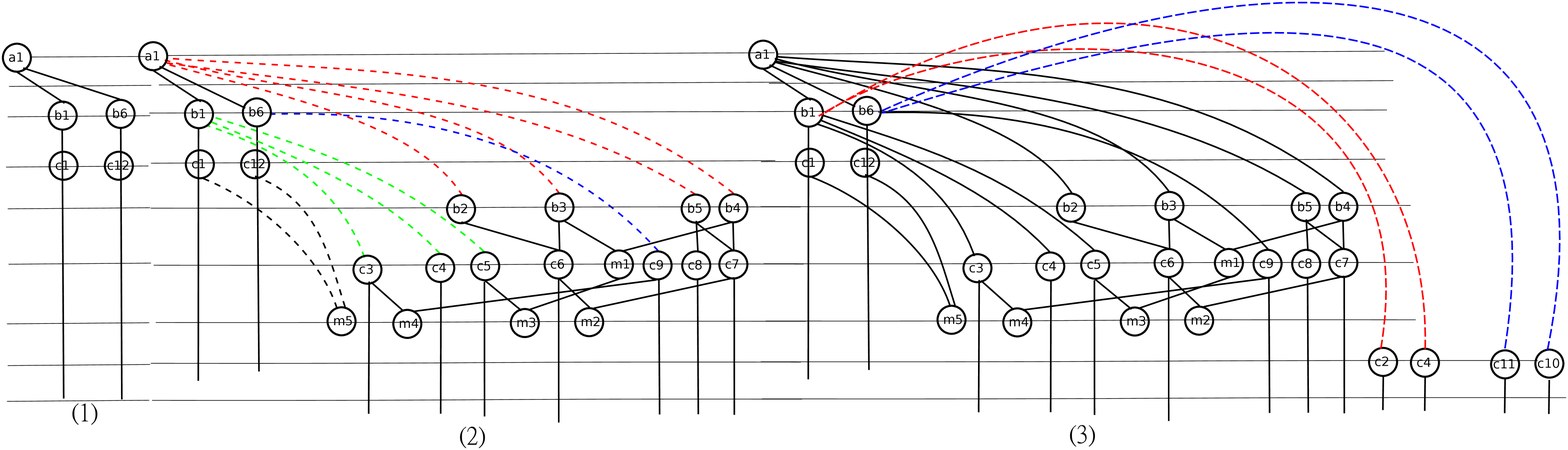}
\caption{(1), (2) and (3) show an example to place a region $\WW$ and $\WW$'s skeleton
of Fig. \ref{fig:skeleton} in a ladder $\HH$ by Algorithm \ref{alg:framework}.
The $\WW$'s skeleton is placed at right of the region $\WW$ in $\HH$ and
starts to place $\WW$'s skeleton from the track three
in $\HH$.
Because $\WW$'s distance-number is one, place the $\WW$'s skeleton from the track three
doesn't make any $X$-crossing edge
between $\WW$ and $\WW$'s skeleton.
After the $\WW$'s skeleton is placed,
the subgraph inside each region of $\WW$'s skeleton
is placed from the track five.
Moreover,
the placement cannot
make $X$-crossing edges with the previous one
in a ladder $\HH$.
}
\label{fig:framework}
\vspace{-0.2in}
\end{figure}

Firstly, this framework place $\WW^M$ on the contiguous tracks from the first track in $\HH$.
Next, this framework also consists of a loop and each iteration of the loop in Algorithm \ref{alg:framework}
starts from the first region $\WW$ rooted at $r$ of the first-in-first-out queue $\tilde{\YY}$
and executes the following steps:
find sequential skeletons $(\Psi_1, \Psi_2, \cdots)$ such that
each skeleton $\Psi_i, i\geq 1,$ roots at the vertex $r(\WW)$ and
each $r(\WW)$'s child inside $\WW$ is a vertex in a skeleton of $(\Psi_1, \Psi_2, \cdots)$.
Also, the sequential skeletons $(\Psi_1, \Psi_2, \cdots)$ are placed orderly at the rightmost part in $\HH$
and starts from the track $(\LL_{\HH}(r(\WW))+2\ZZ)$ in $\HH$.

Before we prove the correctness of Algorithm \ref{alg:framework} in Subsection \ref{subsec:correctness-framework},
we assume the following conjecture in advance.
This conjecture is proved in Section \ref{sec:skeleton}.
\begin{conjecture}\label{conj:skeleton}
Given a region $\WW$ rooted at a vertex $r$,
we have a skeleton $\Psi$ that
$\Psi$ can have sequential regions $\tilde{\WW}(\Psi)$. And, the skeleton $\Psi$ can be placed in $\HH$ as $\Psi_{\HH}$ such that
$\Psi_{\HH}$ is $(\QQ, \XX, \JJ)$-well-placed in $\HH$.
\end{conjecture}

\subsection{Sequential Skeletons $(\Psi_1, \Psi_2, \cdots)$ inside a Region $\WW$ Rooted at a Vertex $r$}\label{subsec:correctness-framework}

In this subsection, we start to show two lemmas, the first one shows how to find sequential skeletons $(\Psi_1, \Psi_2, \cdots)$ inside a region $\WW$ rooted at $r$ such that
each $r$'s child in $\WW$ is a vertex of a skeleton of $(\Psi_1, \Psi_2, \cdots)$,
And the second one shows how to place the sequential skeletons
$(\Psi_1, \Psi_2, \cdots)$ in $\HH$ and
$(\Psi_1, \Psi_2, \cdots)$ can partition the region $\WW$ into sequential regions such that
$(\Psi_1, \Psi_2, \cdots)$ are $(\QQ, \XX, \JJ)$-well-placed in $\HH$.

In the following lemma, we give a constructive proof to find sequential skeletons $(\Psi_1, \Psi_2, \cdots)$ consisting of all $r$'s children inside a region $\WW$.

\begin{lemma}\label{lem:partition}
For a region $\WW=(\BB^L, \BB^R)$ rooted at a vertex $r$, sequential skeletons $\tilde{\Psi}(\WW)=$
$(\Psi_1, \Psi_2, \cdots)$ can be constructed such that
each $r$'s child is a vertex of some skeleton $\Psi$ of $\tilde{\Psi}(\WW)$.
\end{lemma}

\begin{proof}
From Theorem \ref{thm:skeleton}, there exists a skeleton $\Psi$ for the region $\WW$ and
$\Psi$ can partition the region $\WW$ into  sequential regions
$\tilde{\WW}(\Psi)=(\WW_1, \WW_2, \cdots)$.
From Definition \ref{def:skeleton}, each $r$'s children not included into $\Psi$ is inside in a region of $\tilde{\WW}(\Psi)$.
Let $(\WW^M_1, \WW^M_2, \cdots, \WW^M_m)$ be the maximal subsequential regions of
$\tilde{\WW}(\Psi)$ such that each $r$'s child not included in $\Psi$ is inside
a region of $(\WW^M_1, \WW^M_2, \cdots, \WW^M_m)$ rooted at the vertex $r$.

Now for each region $\WW^M_i\in (\WW^M_1, \WW^M_2, \cdots, \WW^M_m)$,
a skeleton $\Psi_i$ for the region $\WW_i$ can be found.
Then by the above discussion,
for each skeleton $\Psi_i, 1\leq i\leq m$, the skeleton $\Psi_i$ can partition the region $\WW^M_i$
into sequential regions $\tilde{\WW}(\Psi_i)$. And,
each $r$'s children inside $\WW^M_i$ not included into $\Psi_i$ is inside
a region of $\tilde{\WW}(\Psi_i)$.
the same partition can be repeatedly executed till each $r$'s child is included into a skeleton.
Hence we can conclude that given a region $\WW$ rooted at a vertex $r$,
we can have sequential skeletons $\tilde{\Psi}(\WW)$
such that each $r$'s child is a vertex of some skeleton $\Psi$ in $\tilde{\Psi}(\WW)$.
\end{proof}

Consider sequential regions $(\WW_1, \WW_2, \cdots)$ in $\HH$.
The sequential regions $(\WW_1, \WW_2, \cdots)$ and the sequential subgraphs $(\tilde{\Psi}_{\HH}(\WW_1), \tilde{\Psi}_{\HH}(\WW_2), \cdots)$ are called an \emph{ordered} layout in $\HH$ if the sequential regions $(\WW_1, \WW_2, \cdots)$ are at left of the sequential subgraphs $(\tilde{\Psi}_{\HH}(\WW_1), \tilde{\Psi}_{\HH}(\WW_2), \cdots)$ in $\HH$, and for any two regions $\WW_i$ and $\WW_j$ in $(\WW_1, \WW_2, \cdots)$, the region $\WW_i$ is at left of the region $\WW_j$ if and only if the subgraph $\tilde{\Psi}_{\HH}(\WW_i)$ is at left of the subgraph $\tilde{\Psi}_{\HH}(\WW_j)$ in $\HH$.
The task of this subsection is to show that if the sequential regions $(\WW_1, \WW_2, \cdots)$ and the sequential subgraphs $(\tilde{\Psi}_{\HH}(\WW_1), \tilde{\Psi}_{\HH}(\WW_2), \cdots)$ in $\HH$ are an ordered layout in $\HH$, then the layout is also $(\QQ, \XX, \DD)$-well-placed in $\HH$.
Now consider sequential edges $(\EE_1, \EE_2, \cdots)$ where each $\EE_i, i\geq 1,$ are edges connected between the region $\WW_i$ and the subgraph $\tilde{\Psi}_{\HH}(\WW_i)$.

Now we prove that the sequential skeletons $\tilde{\Psi}(\WW)$
can be $(\QQ, \XX, \JJ)$-well-placed on contiguous tracks in $\HH$.
From Theorem \ref{thm:skeleton}, we can have a skeleton $\Psi$ for the region $\WW$
such that
$\Psi$ can be $(\QQ, \XX, \JJ)$-well-placed in $\HH$ as $\Psi_{\HH}$ and
have sequential regions $\tilde{\WW}_{\HH}(\Psi)$ in $\HH$.
Then we can pick the maximum subsequential regions $(\WW^M_1, \WW^M_2, \cdots, \WW^M_m)$ of $\tilde{\WW}_{\HH}(\Psi)$
such that each $r$'s child is inside a region of $(\WW^M_1, \WW^M_2, \cdots, \WW^M_m)$.
And, for the sequential regions $(\WW^M_1, \WW^M_2, \cdots, \WW^M_m)$,
we can have corresponding sequential skeletons $(\Psi_1, \Psi_2, \cdots, \Psi_m)$
such that $(\Psi_1, \Psi_2, \cdots, \Psi_m)$ can be $(\QQ, \XX, \JJ)$-well-placed as
$(\tilde{\WW}_{\HH}(\Psi_1), \tilde{\WW}_{\HH}(\Psi_2), \cdots, \tilde{\WW}_{\HH}(\Psi_m))$ in $\HH$.
Now we have placed $(\Psi, \Psi_1, \Psi_2, \cdots, \Psi_m)$ in $\HH$ as sequential regions
$(\tilde{\WW}_{\HH}(\Psi),$
$\tilde{\WW}_{\HH}(\Psi_1),$
$\tilde{\WW}_{\HH}(\Psi_2),$
$\cdots,$
$\tilde{\WW}_{\HH}(\Psi_m))$.
Since regions in $(\tilde{\WW}_{\HH}(\Psi_1),$ $\tilde{\WW}_{\HH}(\Psi_1),$
$\cdots,$
$\tilde{\WW}_{\HH}(\Psi_m))$ are mutually disjoint and
$(\WW^M_1,$
$\WW^M_2,$
$\cdots,$
$\WW^M_m,$
$\tilde{\WW}_{\HH}(\Psi_1),$
$\tilde{\WW}_{\HH}(\Psi_1),$
$\cdots,$
$\tilde{\WW}_{\HH}(\Psi_m))$ are an ordered layout,
edges $\EE_i$ between $\WW^M_i$ and $\tilde{\WW}_{\HH}(\Psi_i)$ and
edges $\EE_j$ between $\WW^M_j$ and $\tilde{\WW}_{\HH}(\Psi_j)$ don't nest in $\HH$.
Hence the layout in $\HH$ is $(\QQ, \XX, \JJ)$-well-placed.
We can repeat the above steps until each $r$'s child inside the region $\WW$ is a vertex of some skeleton of
the sequential skeletons $\tilde{\Psi}(\WW)$ and have sequential regions $\tilde{\WW}_{\HH}(\tilde{\Psi}(\WW))$ in $\HH$.
From the above discussion, we immediately have the following lemma.

\begin{lemma}\label{lem:skeleton-wrap}
For a region $\WW$ rooted at a vertex $r$,
sequential skeletons $\tilde{\Psi}(\WW)$ can be placed as
the new order $\tilde{\Psi}_{\HH}(\WW)$
on contiguous tracks in $\HH$
such that $\tilde{\Psi}_{\HH}(\WW)$ are $(\QQ, \XX, \JJ)$-well-placed in $\HH$ where
the three numbers $\QQ$, $\XX$ and $\JJ$ are constant-bound.
\end{lemma}

In the next lemma,
we prove that when a sequential regions $\tilde{\WW}_{\HH}(\tilde{\Psi}(\WW))$ are placed in $\HH$ in Algorithm \ref{alg:framework},
all edges connecting between $\WW$ and $\tilde{\WW}_{\HH}(\tilde{\Psi}(\WW))$ cannot make $X$-crossing with the existing layout in $\HH$.

\begin{lemma}\label{lem:noncrossing}
Given sequential regions $(\WW_1, \WW_2, \cdots)$ in $\HH$
where each region $\WW_i, i\geq 1,$ roots at a vertex $r_i$.
If each sequential regions $\tilde{\WW}_{\HH}(\tilde{\Psi}(\WW_i)), i\geq 1,$ is placed at right of $(\WW_1, \WW_2, \cdots,$ $\tilde{\WW}_{\HH}(\tilde{\Psi}(\WW_1)),$ $\cdots,$ $\tilde{\WW}_{\HH}(\tilde{\Psi}(\WW_{i-1}))$ and
from the track $\LL_{\HH}(r_i)+2\ZZ$ in $\HH$,
then (1): the edges set
$\EE_i$ connecting between $\WW_i$ and $\tilde{\WW}_{\HH}(\tilde{\Psi}(\WW_i))$ don't have $X$-crossing edges with the layout
$(\WW_1, \WW_2, \cdots,$ $\tilde{\WW}_{\HH}(\tilde{\Psi}(\WW_1)),$ $\cdots,$ $\tilde{\WW}_{\HH}(\tilde{\Psi}(\WW_i))$ in $\HH$.
\end{lemma}

\begin{proof}
Because the region $\WW_i$ is placed at right of the sequential regions $(\WW_1, \WW_2, \cdots, \WW_{i-1})$ and $(\WW_1, \WW_2, \cdots, \WW_i)$ are mutually disjoint,
the edges sets $\EE_i$ and $\EE_j, 1\leq i-1$ don't have any $X$-crossing edge in $\HH$
where $\EE_j$ are all edges connecting between $\WW_j$ and $\tilde{\WW}_{\HH}(\tilde{\Psi}(\WW_j))$.

From Lemma \ref{lem:skeleton-wrap}, we can assume that each edge's gap number on each region $\WW_i$ is less than or equal to $\JJ$.
Because the gap numbers of edges in each edges set $\EE_i$ is greater than or equal to $\ZZ$ and $\ZZ$ is greater than $\JJ$,
each region of $(\WW_1, \WW_2, \cdots,$ $\tilde{\WW}_{\HH}(\tilde{\Psi}(\WW_1)),$ $\cdots,$ $\tilde{\WW}_{\HH}(\tilde{\Psi}(\WW_{i-1}))$ and $\EE_j$ cannot make $X$-crossing in $\HH$.
\end{proof}

\begin{lemma}\label{lem:gap_number}
For each edge $e$ placed in Algorithm \ref{alg:framework},
$e$'s gap number is at most $2\ZZ$.
\end{lemma}
\begin{proof}
Initially,
Each edge of the maximum region $\WW^{\MM}$ except the edges connecting to $\WW^{\MM}$'s root can be placed in $\HH$ with gap number = 1.
Each Edge connecting to $\WW^{\MM}$'s root has gap number = $\ZZ$.
Each edge $e$ connecting to $\WW^{\MM}$'s root inside $\WW^{\MM}$ is an
edge of a skeleton of $\tilde{\Psi}(\WW^{\MM})$ and $\tilde{\Psi}(\WW^{\MM})$ is placed in $\HH$ from the track $2\ZZ+1$ in $\HH$. It leads to $e$'s gap number = $2\ZZ$ in $\HH$.

For a region $\WW$ in $\HH$, the region $\WW$ is called \emph{processed} if $\WW$ has been partitioned by sequential skeletons $\tilde{\Psi}(\WW)$ and each edge $e$ connecting to a boundary of $\WW$,
either $e$ is an edge of a skeleton $\Psi$ for $\WW^{\MM}$ or $e$ is inside a region in $\tilde{\WW}(\Psi)$ in $\HH$.

In each iteration of Algorithm \ref{alg:framework}, we pick the leftmost unprocessed region $\WW$ in $\HH$ and place its sequential skeletons $\tilde{\Psi}(\WW)$ into $\HH$. From Lemma \ref{lem:skeleton-wrap}, for each region $\WW' \in \tilde{\Psi}(\WW)$, each edge of $\WW'$ except connecting to $\WW'$'s root has gap number = $\JJ$.
Now for each vertex $v$ on a boundary of the leftmost unprocessed region $\WW$, all edges connecting to the vertex $v$ inside the region $\WW$ are processed by two consecutive steps:
(1): some edges connecting to the vertex $v$ are edges of a skeleton of $\tilde{\Psi}(\WW)$ and have gap number $\ZZ$ in $\HH$. And,
(2): each remaining edge connecting to the vertex $v$ is inside a region $\WW'$ rooted at the vertex $v$ of $\tilde{\WW}(\tilde{\Psi}(\WW))$.
Then the region $\WW'$ has sequential skeletons $\tilde{\Psi}(\WW')$ such that each remaining edge $e'$ is an edge of some skeleton of $\tilde{\Psi}(\WW')$ and has gap number $2\ZZ$.
Hence we can conclude that
each edge placed in Algorithm \ref{alg:framework} has gap number at most $2\ZZ$.
\end{proof}

From Lemmas \ref{lem:partition}, \ref{lem:skeleton-wrap} and \ref{lem:noncrossing},
we know that each current step in Algorithm \ref{alg:framework}, the placement of sequential skeletons for a region $\WW$ in $\HH$ cannot make $\XX$-crossing with the placement of previous steps in $\HH$
because each edge's gap number in each skeleton is at most $\JJ$.
Also, from Lemma \ref{lem:gap_number},
we know that the distance number of the layout in Algorithm \ref{alg:framework} is at most $2\ZZ$.

Hence we can immediately prove that the framework in Algorithm \ref{alg:framework} can place a composite-layerlike graph $\GG$ as an $(\QQ, \XX, \DD=2\ZZ)$-well-placed layout on $2\DD$ tracks of a ladder $\HH$
in Theorem \ref{thm:cons-track}.

\begin{theorem}\label{thm:cons-track}
If Conjecture \ref{conj:skeleton} can be proven,
then by Algorithm \ref{alg:framework}, every composite-layerlike graph $\GG$ can be an $(\QQ, \XX, \DD)$-well-placed layout on $2\DD$ tracks in a ladder $\HH$.
\end{theorem}

\section{From a Plane Graph $G$ To a Composite-Layerlike Graph $\GG$}\label{sec:transformation}

\begin{figure}[t]
\begin{center}
\includegraphics[width=1\textwidth, angle =0]{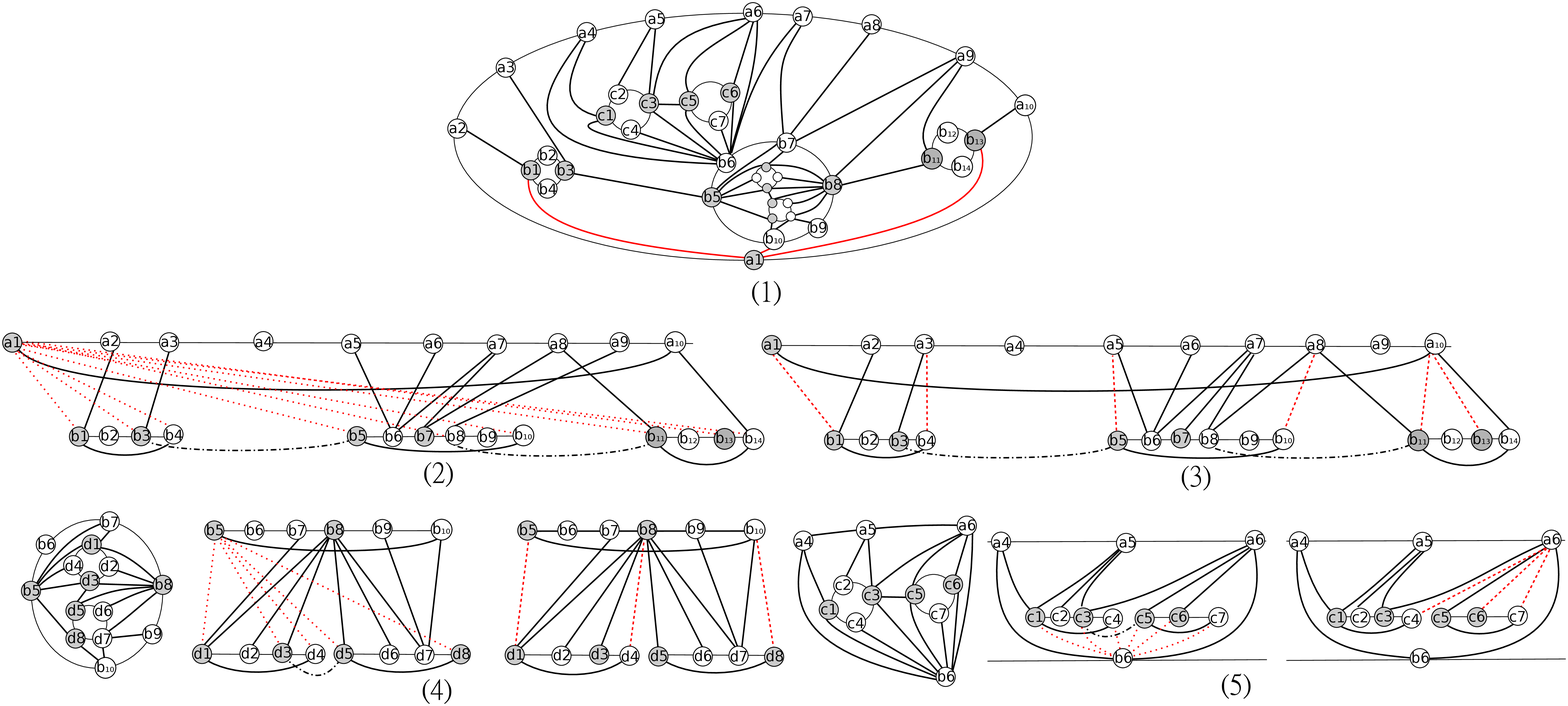}
  \centering
\caption{It shows a procedure to transform a plane graph $G$ of (1) to a composite-layerlike graph $\GG$;
(2) The outer boundary of $G$ is orderly placed from the vertices $a_1$ to $a_{10}$.
And, the second layer consists of three
sequential vertices $(b_1, b_2, b_3, b_4)$,
$(b_5, b_6, b_7, b_8, b_9, b_{10})$ and
$(b_{11}, b_{12}, b_{13}, b_{14})$.
The dash edges are wires.
(3) The wires
$\{(a_1, b_3),$
$(a_1, b_4),$
$(a_1, b_5),$
$(a_1, b_8),$
$(a_1, b_9),$
$(a_1, b_{10}),$
$(a_1, b_{11}),$
$(a_1, b_{13}),$
$(a_1, b_{14}))\}$ and the bridges
$\{(b_3, b_5),$
$(b_8, b_{11})\}$ are removed.
The dummy edges
$\{(a_3, b_4),$
$(a_4, b_5),$
$(a_8, b_{10}),$
$(a_9, b_{11}),$
$(a_9, b_{13})\}$ are added;
(4) shows how to reform a inner cycle into a composite-layerlike graph.
The outer boundary of the inner cycle is placed on a layer in $\GG$
from the vertex $b_5$ and follows the clockwise order as $(b_5, b_6, b_7, b_8, b_9, b_{10})$.
Inside the inner cycle, we have two smaller inner cycles
$(d_1, d_2, d_3, d_4)$ and $(d_5, d_6, d_7, d_8)$.
The two smaller inner cycles are placed in $\HH$ from the vertices $d_1$ and $d_8$ by clockwise order, respectively.
Next, the four wires
$\{(b_5, d_3), (b_5, d_4), (b_5, d_5), (b_5, d_8)\}$ and the bridge $(d_3, d_5)$ are removed,
and the two dummy edges $\{(b_8, d_4), (b_{10}, d_8)\}$ are added into $\GG$;
(5) shows how to reform a down-pointing triangle in a composite-layerlike graph.
The six piles
$\{(b_6, c_1), (b_6, c_3), (b_6, c_4), (b_6, c_5), (b_6, c_6), (b_6, c_7)\}$ and a bridge $(b_7, b_9)$ are removed,
and the dummy edges $\{(a_6, c_4), (a_6, c_6), (a_6, c_7)\}$ are added into a composite-layerlike graph $\GG$.}
\label{fig:reform}
\end{center}
\vspace{-0.2in}
\end{figure}

In this section,
we show how to reform a plane graph $G$ to a composite-layerlike graph $\GG$.
Let $G$ be a plane graph and $\OO(G)$ be its outer boundary.
Each layer in a composite-layerlike graph $\GG$ can be recursively defined as follows:
the first layer is $\OO(G)$ and $\OO_{G}$ is placed as clockwise order $(m, u_1, u_2, \cdots)$.
Then $(\OO_1, \OO_2, \cdots, \OO_p)$ are the sequential maximal inner cycles inside $\OO$ such that
for each maximal inner cycle $\OO_i, 1\leq i\leq p$,
there are some vertices on $\OO_i$ connecting to the vertex $m$.
Now for each maximal inner cycle $\OO_i, 1\leq i\leq p$,
we can walk around the cycle $\OO_i$ by clockwise order
to get two contiguous vertices $(L^U(\OO_i)=(v^i_1=u^i_y, v^i_2, \cdots, v^i_x=u^i_1), L^B(\OO_i)=(u^i_1=v^i_x, u^i_2, \cdots, u^i_y=v^i_1) )$
where each vertex of $L^U(\OO_i)$ don't connect to the vertex $m$ except the first and last vertices $\{v^i_1, v^i_y\}$ of $L^U(\OO_i)$ and each vertex of $L^B(\OO_i)$ connects to the vertex $m$.
All maximal inner cycles $(\OO_1, \OO_2, \cdots, \OO_p)$ can be placed on the second layer as the order:
$(\OO_1, \OO_2, \cdots, \OO_p)$.
Moreover,
for each maximal inner cycle $\OO_i$,
$\OO_i$ is placed on the second layer from the vertex $u^i_1$ by
clockwise order as:
$(v^i_1=u^i_y, v^i_2, \cdots, v^i_x=u^i_y, u^i_{y-1}, \cdots, u^i_2)$.
By the above placement, we can have sequential induced subgraphs
$(G|\triangledown_1, G|\triangledown_2, \cdots, G|\triangledown_q)$
that each induced subgraph $G|\triangledown_i, 1\leq i\leq q$,
is a subgraph induced by a maximal down-pointing triangle $(l_i, \cdots, r_i, m_i)$
where the path from vertices $l_i$ and $r_i$ are on the outer boundary $\OO(G)$ and $m_i$ is a vertex on some inner cycle $\OO_i \in (\OO_1, \OO_2, \cdots, \OO_p)$.
In the followings, we discuss how to place
each induced subgraphs $G|\triangledown_i, 1\leq i\leq q,$
on the subsequent layers in a composite-layerlike graph $\GG$.

For a subgraph induced by a maximal down-pointing triangle $\triangledown=(l, \cdots, r, m)$.
Sequential maximal inner cycles $(\OO_1, \OO_2, \cdots, \OO_q)$ inside $\triangledown=(l, \cdots, r, m)$ can be found
such that each maximal inner cycle $\OO_i, 1\leq i\leq q$ inside $\triangledown$ connects to the vertex $m$ and
each maximal inner cycle $\OO_i$ can be partitioned into two contiguous vertices $(L^U(\OO_i), L^B(\OO_i))$ by clockwise order
where each vertex of $L^U(\OO'_i)$ doesn't connect the vertex $m$ except the first and last vertices of $L^U(\OO_i)$ and
each vertex of $L^B(\OO'_i)$ connects to the vertex $m$.
Also,
(1) the sequential cycles $(\OO_1, \OO_2, \cdots, \OO_q)$ is placed orderly,
(2) the contiguous vertices $\LL^U(\OO_i)$ are placed orderly and
the contiguous vertices $\LL^B(\OO_i)$ are
placed reversely, and
(3) $\LL^U(\OO_i)$ is placed at left of $\LL^B(\OO_i)$
on the subsequent layer of
a new frame $\Pi$ in a composite-layerlike graph $\GG$
where the $\Pi$'s first layer is the same as the down-pointing triangle $\triangledown$'s upper layer.
From the above description,
we assume there is no any chord $(u, v)$ on each cycle $\OO$.
Next we start to explain how to eliminate each chord on each cycle.

Now we plan to remove edges from the above construction such that there is no any $X$-crossing edge in a composite-layerlike graph.
For each maximal inner cycle $\OO$ that the cycle $\OO$ is placed as the clockwise order: $(m, u_1, \cdots)$ on a layerlike graph $\Pi$,
and the sequential maximal inner cycles $(\OO_1, \OO_2, \cdots, \OO_p)$ inside $\OO$ such that
each maximal inner cycle $\OO_i, 1\leq i\leq q,$ connects to the vertex $u_1$.
we remove edges between the vertex $m$ and all vertices on the $(L^B(\OO'_1), L^B(\OO'_2), \cdots, L^B(\OO'_q))$ from the frame $\Pi$.
Moreover,
let $v_i$ be the vertex on the cycle $\OO$ such that
$v_i$ connects to the both cycles $\OO_i$ and $\OO_{i+1}$.
Then, we add each sequential edges between $v_i$ and $L^B(\OO_i), 1\leq i\leq q,$ in the frame $\Pi$.
Similarly, for each maximal down-pointing triangle $\triangledown=(l, \cdots, r, m)$,
we execute the above procedure for the vertex $m$ of the down-pointing triangle $\triangledown$.

During the above transformation, how to process that if there exists a chord
on a cycle $\OO$ is neglected to discuss. The reason is explained below.
A $d$-subdivision of a graph $G$ is a graph obtained by
replacing each edge of $G$ with a path having at most $2 + d$ vertices.
For each chord $(u, v)$ on a cycle $\OO$, a 1-subdivision plane graph $G^1$ without any chord on a cycle $\OO$ can be constructed by the following steps:
(1): find another
vertex $w\notin \OO$ such that the triple vertices $(u, v, w)$ form a triangle face in a plane graph $G$,
(2): a vertex $w'$
can be added inside the face $(u, v, w)$,
(3): the chord $(u, v)$ can be replaced with two edges
$\{(u, w'), (w', v)\}$, and
(4): a dummy edge $(w, w')$ can be added to satisfy triangulation property.
Hence how to process a chord on a cycle $\OO$ found
during the above transformation can be neglected by replaced a chord with two edges.
Now we immediately have the following theorem.

\begin{theorem}\label{thm:reform}
For each plane graph $G$, a $1$-subdivision $G^1$ of $G$ can be reformed into a composite-layerlike graph $\GG$.
\end{theorem}

\section{A Track Layout for a Plane Graph on Constant Number of Tracks}\label{sec:cons-tracks}
We explain how a plane graph $G$ can be placed as a track layout on constant number of tracks.

\begin{theorem}\label{thm:subdivision}\emph{\cite{DW05}}
Suppose a graph $G$ has a $d$-subdivision $k$-track layout.
If the two numbers $k$ and $d$
are constant-bound,
$G$ also have a track layout on constant number of tracks.
\end{theorem}

\begin{theorem}
Every plane graph $G$ has a track layout on constant number of tracks.
\end{theorem}

\begin{proof}
From Theorem \ref{thm:G1-well-placed}, a $1$-subdivision plane graph $G^1$ can have a track layout on constant number of tracks $\HH$.
Because $G^1$'s track number is constant-bound,
by Theorem \ref{thm:subdivision},
$G$ can also have a track layout on constant number of tracks.
\end{proof}

\section{An $(\QQ, \XX, \DD)$-Well-Placed Layout for a Raising Fan $\tilde{\FF}$ in a Ladder $\HH$}\label{sec:layout-fan-path}

\begin{figure}[t]
\begin{center}
\includegraphics[width=1\textwidth, angle =0]{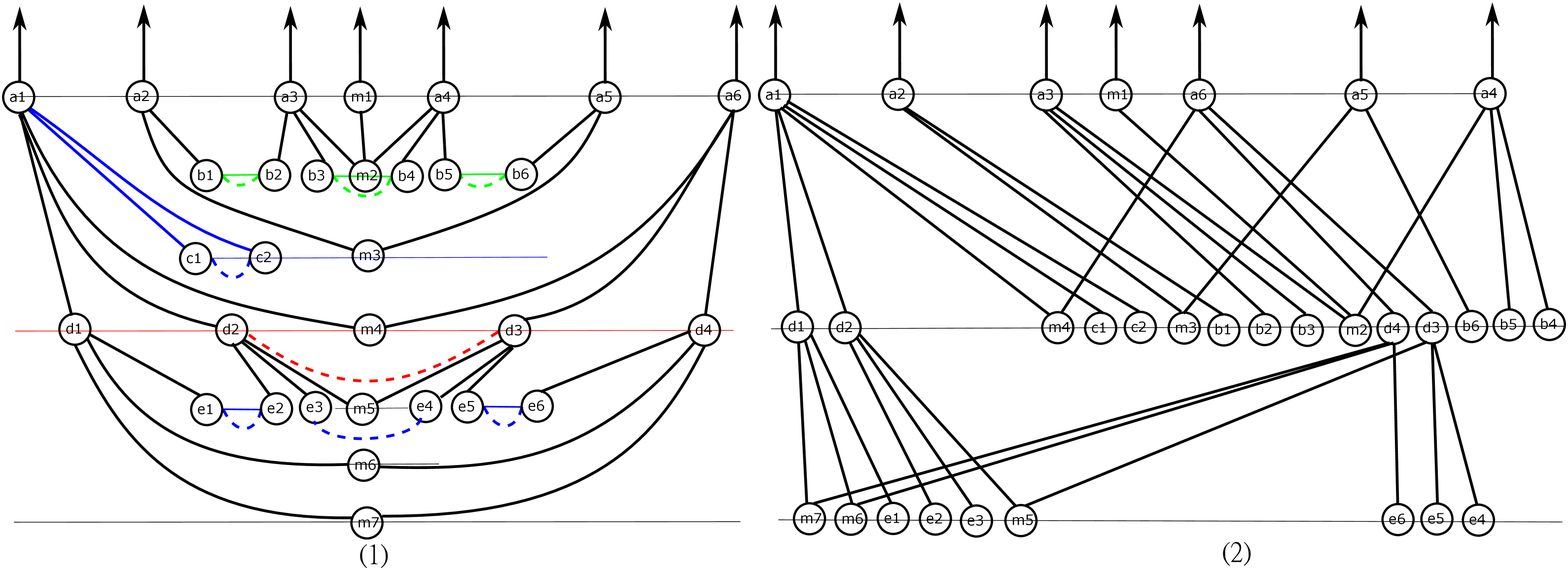}
  \centering
\caption{(1): shows a raising fan $\tilde{\FF}$ consists of six fans
$\FF_1=(a_3, m_1, a_4, m_2)$,
$\FF_2=(a_2, a_5, m_3)$,
$\FF_3=(a_1, a_6, m_4)$,
$\FF_4=(d_2, d_3, m_5)$,
$\FF_5=(d_1, d_4, m_6)$, and
$\FF_6=(d_1, d_4, m_7)$.
Also, $\tilde{\FF}$'s middle path is $\MM=(m_1, m_2, m_3, m_4, m_5, m_6, m_7)$.
For the raising fan $\FF$,
$\LL_1=\{(a_1, a_2, a_3)\}$,
$\RR_1=\{(a_4, a_5, a_6)\}$,
$\MM_1=\{m_1\}$,
$\LL_2=\{(b_1, b_2, b_3), (c_1, c_2), (d_1, d_2)\}$,
$\RR_2=\{(b_4, b_5, b_6), (d_3, d_4)\}$,
$\MM_2=(m_2, m_3, m_4)$,
$\LL_3=\{(e_1, e_2, e_3)\}$,
$\RR_3=\{((e_4, e_5, e_6)\}$,
$\MM_3=(m_5, m_6, m_7)$;
(2): shows an $(\QQ, \XX, \DD)$-well-placed layout for $\tilde{\FF}$ in $\HH$
by orderly placing the left wing of $\tilde{\FF}$ and
reversely placing the right wing of $\tilde{\FF}$.
}
\label{fig:fan-path-layout}
\end{center}
\vspace{-0.2in}
\end{figure}

In this section, we present an approach to have a specified type of sequential fans $\tilde{\FF}$ in a composite-layerlike graph $\GG(\WW^{\MM})$ called a \emph{raising-fan} path and defined later and
show hot to place it as an $(\QQ, \XX, \DD)$-well-placed layout in a ladder $\HH$.
Also, an $(\QQ, \XX, \DD)$-well-placed layout for a \emph{raising-fan} path $\tilde{\FF}$ can partition the region
$\WW^{\MM}$ into sequential regions in $\HH$.

A \emph{fan} $\FF$ consists of sequential vertices $(u_1, \cdots, u_a, m)$
such that $(u_1, \cdots, u_a)$ and $w$ are on two contiguous
layers $i$ and $i+1$ of a layerlike graph, and $w$ connects each vertex in $(u_1, \cdots, u_a$.
The sequential vertices $(u_1, \cdots, u_a)$ are called upper vertices of the fan $\FF$.
Also, the vertex $m$ are called the lower vertex of the fan $\FF$.
A \emph{raising-fan} path $\tilde{\FF}$ consists of sequential fans $(\FF_1, \FF_2, \cdots)$
such that for each fan $\FF_i, i\geq 1$,
there is a down-pointing triangle in $\FF_{i+1}$ bounds $\FF_i$.
Also, the lower vertices of the raising fan $\tilde{\FF}$ form
 sequential vertices $\MM=(m_1, m_2, \cdots)$ and is called a \emph{middle path}
where each vertex $m_i, i\geq 1$, is the lower vertex of the fan $\FF_i$.
For a fan $\FF\in \tilde{\FF}$ with the middle path $\MM$,
edges of the fan $\FF$ at left and right of the middle path $\MM$ are called
a \emph{left arm} and \emph{right arm} of the fan $\FF$, respectively

A path $\Re(v)$ is called a \emph{raising-path} from a vertex $v$ in a composite-layerlike graph $\GG$
if the path $\Re(v)$ starts from the vertex $v$ along edges from the lower layer to the upper layer
till the first layer in $\GG$,.
Any two raising-paths $\Re(v_1)$ and $\Re(v_2)$,
$\Re(v_1)$ and $\Re(v_2)$ are called \emph{upward-merging} if either
the two raising-paths $\Re(v_1)$ and $\Re(v_2)$ are vertex-disjoint or
the intersection of the two raising-paths $\Re(v_1) \cap \Re(v_2)$ is a subpath from some vertex to a vertex on the first layer in $\GG$.
A set of raising-paths $\tilde{\Re}=\{\Re(v)|v\in V(\GG)\}$ in a composite-layerlike graph $\GG$ are called \emph{upward-merging} to $\GG$
if for any two raising-paths $\Re(v_1)$ and $\Re(v_2)$ in $\tilde{\Re}$ are upward-merging and each vertex in $\GG$
is in a raising-path $\Re(v)$ in $\tilde{\Re}$.
From now on, when mention a raising-path $\Re(v)$, it always means that
the path $\Re(v)$ is a raising-path in some specific upward-merging raising-path set $\tilde{\Re}$ for a composite-layerlike graph $\GG$.

Assume that $\triangledown$ is a down-pointing triangle inside a fan $\FF$ of a raising fan $\tilde{\FF}$ with the middle path $\MM$.
\begin{itemize}
\item A \emph{left spine} of $\triangledown$ with respect to  the $\tilde{\FF}$'s middle path $\MM$ consists of
sequential maximal cycles $LS(\triangledown)=(\OO^L_1, \OO^L_2, \cdots, \OO^L_p)$ at left of the middle path $\MM$
such that for each cycle $\OO^L_i, 1\leq i\leq p$, there is an edge connecting between the upper vertices of $\triangledown$ and the boundary of $\OO^L_i$.

\item A \emph{right spine} of $\triangledown$ with respect to  the $\tilde{\FF}$'s middle path $\MM$ consists of
sequential maximal cycles $RS(\triangledown)=(\OO^R_1, \OO^R_2, \cdots, \OO^R_q)$ at right of the middle path $\MM$
such that for each cycle $\OO^R_i, 1\leq i\leq q$, there is an edge connecting between the upper vertices of $\triangledown$ and the boundary of $\OO^R_i$.

\item the $i$-th \emph{joint} of a spine inside $\FF'$ consists of the $i$-th cycle's the leftmost and rightmost vertices in the spine of $\triangledown$.
\end{itemize}
If a down-pointing triangle $\triangledown\in \FF$ which is not passed through the middle path $\MM$, then the sequential maximal cycles $(\OO_1, \OO_2, \cdots, \OO_t)$ is called a \emph{spine} in $\triangledown$
because the spine is not partitioned by the middle path $\MM$.

Let $\FF$ be a fan in a raising fan $\tilde{\FF}$ with the middle path $\MM$ that
Then the fan $\FF$ are partitioned by the middle path $\MM$ into the following sequential down-pointing triangles
$(\triangledown^L_1,$
$\triangledown^L_2,$
$\cdots,$
$\triangledown^L_a),$
$\triangledown^M,$
$(\triangledown^R_1,$
$\triangledown^R_2,$
$\cdots,$
$\triangledown^R_b)$
where each down-pointing triangle $\triangledown^L_i, 1\leq i\leq a,$ is at left of the middle path $\MM$,
the down-pointing triangle $\triangledown^M$ consists of the middle path $\MM$ and
each down-pointing triangle $\triangledown^R_i, 1\leq i\leq b,$ is at right of the middle path $\MM$.


\begin{itemize}
\item A \emph{left wing} $\LL(\FF)$ of a fan $\FF$ is the union of raising-paths consisting of all joints
in the sequential left down-pointing triangles $\triangledown^L_i, 1\leq i\leq a$ of the fan $\FF$,
and the down-pointing triangle $\triangledown^M$'s left joints with respect to the middle path $\MM$. And

\item A \emph{right wing} $\RR(\FF)$ of a fan $\FF$ is the union of raising-paths consisting of all joints
in the sequential right down-pointing triangles $\triangledown^R_i, 1\leq i\leq b$ of the fan $\FF$,
and the down-pointing triangle $\triangledown^M$'s right joints with respect to the middle path $\MM$.
\end{itemize}

From the definition of a raising-path, two contiguous raising-paths $(\Re(u), \Re(u')) \in \LL$,
$\Re(u)$ and $\Re(u')$ form
a disjoint region $\WW^L$, and $\Re(u)$ and $\Re(u')$ are the
left and right boundaries of $\WW^L$, respectively.
Similarly,
two raising-paths $(\Re(u), \Re(u')) \in \RR$ form
a region $\WW^R$, and $\Re(u)$ and $\Re(u')$ are the
left and right boundaries of $\WW^R$, respectively.
Hence we have two sequential regions for a fan $\FF$: $\tilde{\WW}^L(\FF)$ and $\tilde{\WW}^R(\FF)$ that are at left and right of the middle path $\MM$, respectively.

From the definitions of a raising fan $\tilde{\FF}$ and a composite-layerlike graph $\GG$,
we can give another representation for a raising fan $\tilde{\FF}$ as
$(\FF_{1, 1}, \FF_{1, 2}, \cdots, \FF_{1, a_1},$ $\FF_{2, 1}, \FF_{2, 2}, \cdots, \FF_{2, a_2},$ $\cdots)$
where for each subsequential fans $(\FF_{i, 1}, \FF_{i, 2}, \cdots, \FF_{i, a_i}), i\geq 1$,
all upper vertices of each fan $\FF_{i, j}, 1\leq j\leq a_i,$ are on the $i$-th layer of the layerlike graph $\Pi$ of $\GG$.
Also, subsequential fans
$\tilde{\FF}^C=(\FF_{1, 1}, \FF_{1, 2}, \cdots,$ $\FF_{1, a_1}, \FF_{2, 1}, \FF_{2, 2},$ $\cdots, \FF_{2, a_2}, \cdots)$
are called \emph{characteristic-fans} of $\tilde{\FF}$
if for each $i\geq 1$, all upper vertices of $(\FF_{i, 1}, \FF_{i, 2}, \cdots, \FF_{i, a_i})$
are on the $i$-th layer of the $\GG$'s layerlike graph $\Pi$.
Also, the union of left wings and right wings of all fans in a characteristic-raising fan
on the $i$-th layer are denoted as $\LL_i(\tilde{\FF}^C)$ and $\RR_i(\tilde{\FF}^C)$, respectively. And,
the sequential lower vertices $(m_{i, 1}, m_{i, 2}, \cdots, m_{i, a_i})$
of the sequential fans $(\FF_{i, 1}, \FF_{i, 2}, \cdots, \FF_{i, a_i})$
are denoted as $\MM_i(\tilde{\FF}^C)$.

Let $\tilde{\FF}^C$
$=(\FF_{1, 1}, \FF_{1, 2},$
$\cdots,$
$\FF_{1, a_1},$
$\FF_{2, 1},$
$\FF_{2, 2},$
$\cdots,$
$\FF_{2, a_2},$
$\cdots)$ be the characteristic-raising fan of $\tilde{\FF}$.
The characteristic-raising fan $\tilde{\FF}^C$ can form sequential regions as follows:
initially, let the two term $\tilde{\WW}^L(\tilde{\FF}^C)$ and $\tilde{\WW}^R(\tilde{\FF}^C)$ be empty sequence.
Then repeatedly find the last fan $\FF\in \tilde{\FF}^C$
(the last fan $\FF\in \tilde{\FF}^C$ is inside a region between the rightmost and leftmost boundaries of $\tilde{\WW}^L(\tilde{\FF}^C)$ and $\tilde{\WW}^R(\tilde{\FF}^C)$, respectively),
remove the fan $\FF$ from $\tilde{\FF}^C$, and
add two sequential regions $\tilde{\WW}^L(\FF)$ and $\tilde{\WW}^R(\FF)$
into $\tilde{\FF}^L(\tilde{\FF}^C)$ and $\tilde{\FF}^L(\tilde{\FF}^C)$, respectively
where the two sequential regions $\tilde{\WW}^L(\FF)$ and $\tilde{\WW}^R(\FF)$ are formed by the fan $\FF$
that are at left and at right of the middle path $\MM$, respectively.
After all fan are removed from the characteristic-raising fan $\tilde{\FF}^C$,
we can have two sequential regions $\tilde{\WW}^L(\tilde{\FF}^C)$ and
$\tilde{\WW}^R(\tilde{\FF})$ that are at left and right of the middle path $\MM$, respectively.

For each subsequential fans $\tilde{\FF}_{i, j}$ from $\FF_{(i, j)+1}$ to $\FF_{(i, j+1)-1}, i, j\geq1$,
we know that the unions of left wings $\LL_{i, j}$ and right wings $\RR_{i, j}$ for all fans $\FF$ from $\FF_{(i, j)+1}$ to $\FF_{(i, j+1)-1}$
are inside the rightmost region $\WW^L_{i, j}$ and leftmost regions $\WW^R_{i, j}$ of $\tilde{\WW}^L(\tilde{\FF}^C)_{i, j}$ and $\tilde{\WW}^R(\tilde{\FF}^C)_{i, j}$, respectively.
Hence we can recursively partition the regions $\WW^L_{i, j}$ and $\WW^R_{i, j}$ by the two left wings $\LL_{i, j}$ and right wings $\RR_{i, j}$ and replace the two regions
$\tilde{\WW}^L_{i, j}$ and $\tilde{\WW}^R_{i, j}$
by the two sequential regions $\tilde{\WW}^L_{i, j}$
and $\tilde{\WW}^R_{i, j}$, respectively.
Now we can conclude that
the union of the left wings and right wings of the raising fan $\tilde{\FF}$
form two sequential regions $\tilde{\WW}^L(\tilde{\FF})=(\WW^L_1, \WW^L_2, \cdots, \WW^L_p)$ and
$\tilde{\WW}^R(\tilde{\FF})=(\WW^R_1, \WW^R_2, \cdots, \WW^R_q)$.

We can conclude that a raising fan $\tilde{\FF}=(\FF_1, \FF_2, \cdots)$ can form sequential regions as
$(\cdots,$
$\tilde{\WW}^L_2,$ $\tilde{\WW}^L_1,$ $\tilde{\WW}^R_1,$ $\tilde{\WW}^R_2,$
$\cdots)$
where $\tilde{\WW}^L_i$
and $\tilde{\WW}^R_i$
are sequential regions at left and right of the middle path $\MM$ for a fan $\FF_i$ of a raising fan $\tilde{\FF}$, respectively.
From the above discussion, we immediately have the following lemma.

\begin{lemma}\label{lemma:fan-path-partition}
Given a raising fan $\tilde{\FF}$ and its middle path $\MM$,
the unions of $\tilde{\FF}$'s left wings and right wings can form sequential regions $(\WW^L_1, \WW^L_2, \cdots, \WW^L_p,$ $\WW^R_1, \WW^R_2, \cdots, \WW^R_q)$
where the two subsequential regions $(\WW^L_1, \WW^L_2,$ $\cdots,$ $\WW^L_p)$ and $(\WW^R_1, \WW^R_2,$ $\cdots,$ $\WW^R_q)$ are
at left and right of the middle path $\MM$, respectively.
Moreover, $(\WW^L_1, \WW^L_2, \cdots, \WW^L_p)$ and $(\WW^R_1, \WW^R_2, \cdots, \WW^R_q)$
are called left and right sequential regions of $\tilde{\FF}$, and denoted to $\tilde{\WW}^L(\tilde{\FF})$ and $\tilde{\WW}^R(\tilde{\FF})$, respectively.
\end{lemma}

Now we can consider how to have an $(\QQ, \XX, \DD)$-well-placed layout for a raising fan $\tilde{\FF}=(\FF_1, \FF_2, \cdots)$ in a ladder $\HH$.
When we have a raising fan $\tilde{\FF}$, it implies that we have two sequential
regions $\tilde{\WW}^L(\tilde{\FF})$ and $\tilde{\WW}^R(\tilde{\FF})$ from Lemma \ref{lemma:fan-path-partition}.
The basic idea to have an $(\QQ, \XX, \DD)$-well-placed layout in $\HH$ is that
we orderly place the sequential regions $\tilde{\WW}^L(\tilde{\FF})$ and reversely place the sequential regions $\WW^R(\tilde{\FF})$ in $\HH$.
To place the left wing and right wing of a raising fan $\tilde{\FF}$,
we can start to place the $\tilde{\FF}$'s characteristic-fans
$(\FF_{1, 1}, \FF_{1, 2}, \cdots, \FF_{1, a_1},$ $\cdots,$ $\FF_{2, 1}, \FF_{2, 2}, \cdots, \FF_{2, a_2}, \cdots)$
and their union of left wings $\LL_i$, union of right wings $\RR_i$ and
their sequential lower vertices $\MM_i=(m_{i, 1}, m_{i, 2}, \cdots, m_{i, a_i})$ as the order:
for each $i\geq 1$, we place $\LL_i$ orderly,
$\MM_i$ reversely and $\RR_i$ reversely on the $i$-th layer of $\HH$.

Since the sequential regions $\tilde{\WW}^L(\tilde{\FF}^C_i)$ are placed orderly in $\HH$,
we must place $\MM_i(\tilde{\FF}^C)=(m_{i, 1},$ $m_{i, 2},$ $\cdots,$ $m_{i, a_i})$ reversely in $\HH$ as
$(m_{i, a_i},$ $\cdots,$ $m_{i, 2},$ $m_{i, 1})$ to avoid $X$-crossing edges in $\HH$.
For each raising fan $\tilde{\FF}^C_i=(\FF^C_{i, 1}, \FF^C_{i, 2}, \cdots, \FF^C_{i, a_i})$, the sequential lower vertices $\MM_i=(m_{i, 1}, m_{i, 2}, \cdots, m_{i, a_i})$ are reversely placed as
$(m_{i, a_i},$ $\cdots,$ $m_{i, 2},$ $m_{i, 1})$ in $\HH$,
the sequential regions $\tilde{\WW}^R(\tilde{\FF}^C_i)$ need to be placed reversely in $\HH$ to avoid that
$\tilde{\WW}^R(\tilde{\FF}^C_i)$ make $X$-crossing edges in $\HH$.

\begin{lemma}\label{lem:fan-path-region-order}
Given the characteristic-raising fan $\tilde{\FF}^C$ of a raising fan $\tilde{\FF}$,
the unions of left wings and right wings of $\tilde{\FF}^C$
form two sequential regions $\tilde{\WW}^L(\tilde{\FF}^C)$ and
$\tilde{\WW}^R(\tilde{\FF}^C)$,
the sequential regions $\tilde{\WW}^L(\tilde{\FF}^C)$ are orderly placed
as $\tilde{\WW}^L_{\HH}(\tilde{\FF}^C)$ and
the sequential regions
$\tilde{\WW}^R(\tilde{\FF}^C)$ are reversely placed
as $\tilde{\WW}^R_{\HH}(\tilde{\FF}^C)$.
Also,
the sequential regions $\tilde{\WW}_{\HH}(\tilde{\FF}^C)=(\tilde{\WW}^L_{\HH}(\tilde{\FF}^C),
\tilde{\WW}^R_{\HH}(\tilde{\FF}^C))$ are $(\QQ, \XX, \DD)$-well-placed in $\HH$.
\end{lemma}

In the next theorem, we prove that a raising fan $\tilde{\FF}$ can be $(\QQ, \XX, \DD)$-well-placed in a ladder $\HH$.

\begin{theorem}\label{thm:fan-path-region-order}
Given a raising fan $\tilde{\FF}$, the left wings and right wings of $\tilde{\FF}$
form two sequential regions $\tilde{\WW}^L(\tilde{\FF})=(\WW^L_1, \WW^L_2, \cdots, \WW^L_p)$ and
$\tilde{\WW}^R(\tilde{\FF})=(\WW^R_1, \WW^R_2, \cdots, \WW^R_q)$.
If the sequential regions $\tilde{\WW}^L(\tilde{\FF})=(\WW^L_1, \WW^L_2, \cdots, \WW^L_p)$ are sequentially placed in $\HH$ as the order:
$\tilde{\WW}^L_{\HH}(\tilde{\FF})=(\WW^L_1, \WW^L_2, \cdots, \WW^L_p)$
and
the sequential regions
$\tilde{\WW}^R(\tilde{\FF})=(\WW^R_1, \WW^R_2, \cdots, \WW^R_q)$ are reversely placed
as the order in $\HH$:
$\tilde{\WW}_{\HH}(\tilde{\FF})=$
$(\WW^L_1, \WW^L_2,$ $\cdots, \WW^R_p,$
$\WW^R_q, \cdots,$ $\WW^R_2, \WW^R_1)$,
then the sequential regions $\tilde{\WW}_{\HH}(\tilde{\FF})=(\tilde{\WW}^L_{\HH}(\tilde{\FF}), \tilde{\WW}^R_{\HH}(\tilde{\FF}))$ are $(\QQ, \XX, \DD)$-well-placed in $\HH$.
\end{theorem}

\begin{proof}
From Lemma \ref{lemma:fan-path-partition}, we know that
a raising fan $\tilde{\FF}^C$ can form two sequential regions $(\tilde{\WW}^L(\tilde{\FF}),$
$\tilde{\WW}^R(\tilde{\FF}))$
that are at left and right of the middle path $\MM$, respectively.

For a characteristic-raising fan $\tilde{\FF}^C$,
we know that
$(\tilde{\WW}^L_{\HH}(\tilde{\FF}^C), \tilde{\WW}^L_{\HH}(\tilde{\FF}^C))$ are $(\QQ, \XX, \DD)$-well-placed in a ladder $\HH$ by Lemma \ref{lem:fan-path-region-order}.
Moreover,
the sequential regions
$\tilde{\WW}^L(\tilde{\FF}^C)$ are orderly placed in $\HH$ as$\tilde{\WW}^L_{\HH}(\tilde{\FF}^C)$ and
the sequential regions
$\tilde{\WW}^R(\tilde{\FF}^C)$ are reversely placed in $\HH$ as $\tilde{\WW}^R_{\HH}(\tilde{\FF}^C)$.
For each subsequential fans $\tilde{\FF}_{i, j}$ from $\FF_{(i, j)+1}$ to $\FF_{(i, j+1)-1}, i, j\geq1$,
we know that the unions of left wings $\LL_{i, j}$ and right wings $\RR_{i, j}$ of all fans from $\FF_{(i, j)+1}$ to $\FF_{(i, j+1)-1}$
are inside the rightmost region $\WW^L_{i, j}$ and leftmost regions $\WW^R_{i, j}$ of $\tilde{\WW}^L(\tilde{\FF}^C)_{i, j}$ and $\tilde{\WW}^R(\tilde{\FF}^C)_{i, j}$, respectively.
Hence we can recursively partition the regions $\WW^L_{i, j}$ and $\WW^R_{i, j}$ by the two left wings $\LL_{i, j}$ and right wings $\RR_{i, j}$ to
have two sequential regions
$\tilde{\WW}(\LL_{i, j})$ and
$\tilde{\WW}(\RR_{i, j})$ that can be
$(\QQ, \XX, \DD)$-well-placed as
$\tilde{\WW}_{\HH}(\LL_{i, j})$ and
$\tilde{\WW}_{\HH}(\RR_{i, j})$
in $\HH$.
Because the regions $\WW^L_{i, j}$ and $\WW^R_{i, j}$ are orderly and reversely placed in $\HH$, respectively,
we can place the sequential regions $\tilde{\WW}_{\HH}(\LL_{i, j})$ orderly
and the sequential regions $\tilde{\WW}_{\HH}(\RR_{i, j})$ reversely
inside the two regions
$\tilde{\WW}^L_{i, j}$ and $\tilde{\WW}^R_{i, j}$, respectively
to have an $(\QQ, \XX, \DD)$-well-placed layout in $\HH$.
Moreover, the middle path from $\FF_{(i, j)+1}$ to $\FF_{(i, j+1)-1}$ is placed reversely between the two vertices $m_{(i, j)}$ and
$m_{(i, j+1)-1}$.

Because (1) $\LL_{i, j}$ are placed between $\FF_{i, j}$ and $\FF_{i, j+1}$, and
(2) the middle path from $\FF_{(i, j)+1}$ to $\FF_{(i, j+1)-1}$ is placed reversely between the two vertices $m_{(i, j)}$ and
$m_{(i, j+1)-1}$,
each left arm from
$\FF_{(i, j)+1}$ to $\FF_{i, j+1}-1$ is placed between tracks $i$ and $i+1$ cannot make $X$-crossing with
left arms of the characteristic-raising fan $\tilde{\FF}^C$ between tracks $i$ and $i+1$in $\HH$.
Similarly, each right arm of $\RR_{i, j}$ is placed between
tracks $i$ and $i+1$ cannot
make $X$-crossing with
right arms of the characteristic-raising fan $\tilde{\FF}^C$ between tracks $i$
and $i+1$ in $\HH$.

Finally, we conclude that
the sequential regions $\tilde{\WW}_{\HH}(\tilde{\FF})=(\tilde{\WW}^L_{\HH}(\tilde{\FF}), \tilde{\WW}^R_{\HH}(\tilde{\FF}))$ are $(\QQ, \XX, \DD)$-well-placed in $\HH$.

\end{proof}

\section{How to Find a Skeleton $\Psi$ in a Region $\WW$ to Have an $(\QQ, \XX, \DD)$-Well-Placed Layout $\Psi_{\HH}$ in a Ladder $\HH$?}\label{sec:skeleton}

In this section, we prove that a skeleton $\Psi$ of a regions $\WW$ can be obtained from all \emph{rightward-outer} and \emph{leftward-outer} fans in $\WW$ and can have an
$(\QQ, \XX, \DD)$-well-placed layout in $\HH$.

\subsection{A Forest-Like Representation $\clubsuit_{\BB}$ of Raising Fans from a Boundary $\BB$}\label{sec:forest-like}

In this subsection, we show how to build a forest-like representation $\clubsuit_{\BB}$
from all \emph{rightward-outer} fans of all vertices on a boundary $\BB$.

Given a region $\WW=(\BB^L, \BB^R)$,
let $v$ be a vertex on the boundary $\BB^L$ and $(v_1, v_2, \cdots)$ be the sequential children of $v$ inside the region $\WW$.
A fan $\FF=(u_1, u_2, \cdots, m)$ is called a \emph{rightward} for $v$' children on a left boundary of a region $\WW'$
if $\FF$'s upper vertices can be
partitioned into two contiguous subsequences that the first contiguous subsequence is
among $v$'s children and
the second contiguous subsequence
is not among $v$'s children.

A fan $\FF=(u_1, u_2, \cdots, m)$ is called a \emph{rightward} fan for a vertex $v$ on a left boundary of a region $\WW'$
if $\FF$'s the first upper vertex $u_1$ is the vertex $v$.

Given a subregion $\WW'$ inside a region $\WW$ and the vertex $v$ is on the lowest layer of the intersection of the $\WW$'s left boundary $\BB^L$ and $\WW'$'s left boundary,
a \emph{leftmost-raising-fan} path $\tilde{\FF}$ inside the subregion $\WW'$ of $\WW$ is a maximal raising fan consisting of all $v$'s rightward-outer fans.
Intuitively, the layer of any rightward-outer fan not included in $\tilde{\FF}$ can not be lower than the lowest layer of a rightward-outer fan included in $\tilde{\FF}$.

Similarly,
a \emph{leftward-outer} fan and
a \emph{leftward} fan of a vertex $v$ in a region $\WW'$ also can be defined symmetrically to a rightward fan and a rightward-outer fan of a vertex $v\in \BB^R$, respectively. And,
a \emph{rightmost-raising-fan} path inside a region $\WW'$ also can be defined symmetrically to a leftmost raising fan inside a region $\WW'$.

Without loss of generality,
assume that $\BB$ is $\WW$'s
the left boundary,
we plan to find a set of raising fans from $\BB$'s all rightward-outer fans and
represent them as a forest-like representation $\clubsuit_{\BB}$
where each vertex of $\clubsuit_{\BB}$ represents a raising fan.
Initially $\clubsuit_{\BB}$'s root is a raising fan $\tilde{\FF}$ consisting the only region $\WW$
and add the region $\WW$ into $\tilde{\WW}(\clubsuit_{\BB})$.
Then $\clubsuit_{\BB}$ can be constructed as follows:
repeatedly find a region $\WW'$ from the sequential regions having uncolored rightward-outer fans inside the region $\WW'$
and $\WW'$ is a region of a raising fan $\tilde{\FF}$
, to
\begin{enumerate}
\item find a maximal raising fan $\tilde{\FF}'$ consisting of all leftmost and uncolored fans inside the region $\WW'$,

\item  assign the raising fan $\tilde{\FF}'$ to a child of $\tilde{\FF}$ consisting of the region $\WW'$ in $\clubsuit_{\BB}$ and color the raising fan $\tilde{\FF}'$,

\item partition the region $\WW'$ into sequential regions $\tilde{\WW}'$ by the raising fan $\tilde{\FF}'$ and replace the region $\WW'$ by the sequential regions $\tilde{\WW}'$ in $\tilde{\WW}(\clubsuit_{\BB})$, and

\end{enumerate}
till no any such region can be found. Now we immediately have the following two lemmas:




\begin{lemma}
Given a region $\WW=(\BB^L, \BB^R)$ with the left and right boundaries $\BB^L$ and $\BB^R$, the collection of all rightward-outer and leftward-outer fans of the boundary $\BB^L$ and
$\BB^R$ form sequential regions $(\WW^L_1, \WW^L_2, \cdots, \WW^L_p)$ and $(\WW^R_1, \WW^R_2, \cdots, \WW^R_q)$, respectively.
Also, the two sequential regions partition the region $\WW$ into sequential regions
$(\WW^L_1, \WW^L_2, \cdots, \WW^L_p,$ $\WW^M,$ $\WW^R_1, \WW^R_2, \cdots, \WW^R_q)$ where $\WW^M$
is a region bounded by the right boundary of $\WW^L_p$ and the left boundary of $\WW^R_1$.
\end{lemma}

\begin{lemma}\label{lem:forest-like}
Given a region $\WW=(\BB^L, \BB^R)$,
we can have two forest-like representations $\clubsuit_{\BB^L}$ and $\clubsuit_{\BB^R}$
from the collection of all rightward-outer and leftward-outer fans of $\BB^L$ and $\BB^R$, respectively.
\end{lemma}
A forest-like structure $\clubsuit_{\BB}$ is called \emph{left-forest-like} structure, if $\BB$ is the left boundary of a region $\WW$.

Recall that a boundary $\BB$'s forest-like structure $\clubsuit_{\BB}$ is the union of all rightward-outer fans of $(u_1, u_2, \cdots, u_l)$
where each $u_i, i\geq 1,$ is a vertex on the boundary $\BB$.
Observe that for each vertex $u\in \BB$,
every $u$'s child $u'$ is inside a region of $u$'s rightward-outer fan $\FF$ and $\FF$ is a fan in a raising fan $\tilde{\FF}\in (\tilde{\FF}_1, \tilde{\FF}_2, \cdots)$.
Hence $u'$ is inside a region of $\tilde{\WW}(\tilde{\FF})$ and $\tilde{\WW}(\clubsuit_{\BB})$.
From the above discussion, we have the following lemma:

\begin{lemma}\label{lem:region-partition}
Given a region $\WW=(\BB^L, \BB^R)$ and a boundary $\BB\in (\BB^L, \BB^R)$, for each vertex $u\in \BB$, every $u$'s child is inside a region of sequential regions $\tilde{\WW}(\clubsuit_{\BB})$.
\end{lemma}

From the above lemma, we know that the union of the two forest-like structures
$\clubsuit_{\BB^L}$ and $\clubsuit_{\BB^R}$ form a skeleton of a region $\WW$.
$\clubsuit_{\BB^L}$ and $\clubsuit_{\BB^R}$ are called \emph{left skeleton} $\Psi^L(\WW)$ and \emph{right skeleton} $\Psi^R(\WW)$, respectively.

\begin{theorem}\label{thm:skeleton}
Given a region $\WW=(\BB^L, \BB^R)$, the union of the left skeleton and right skeleton $\Psi^L(\WW)$ and $\Psi^R(\WW)$ is a skeleton of $\WW$.
\end{theorem}

\subsection{An $(\QQ, \XX, \DD)$-Well-Placed Layout in a Ladder $\HH$ for a Forest-Like Structure $\clubsuit_{\BB}$}\label{sec:placement-forest-like}

\begin{algorithm}[t]
\caption{A Framework to Place Forest-Like Raising-Fan Paths $\clubsuit_{\BB}$ in a Ladder $\HH$}
\label{alg:framework-boundary}

\KwIn{A forest-like structure $\clubsuit_{\BB}$ consists of the only one root $r$.}
\KwOut{sequential regions $\tilde{\WW}_{\HH}(r)$ in $\HH$.}

Let $\tilde{\FF}$ be the raising fan at the $\clubsuit_{\BB}$'s root $r$\;

Orderly place the sequential regions $\tilde{\WW}_{\HH}(\tilde{\FF})$ in $\HH$\;

Orderly add the sequential regions $\tilde{\WW}_{\HH}(\tilde{\FF})$ into the first-in-first-out queue $\tilde{\WW}$\;

\While{$\tilde{\WW}$ is not empty}
{

Remove the first region $\WW'$ in $\tilde{\WW}$\;

Place the subsequential regions
$\tilde{\WW}'_{\HH}$ of
$\tilde{\WW}_{\HH}(\clubsuit_{\BB}(\WW'))$ at rightmost in $\HH$
where each region $\WW''$ in $\tilde{\WW}'_{\HH}$ consists of some raising fans in $\clubsuit_{\BB}$\;

Orderly add $\tilde{\WW}_{\HH}(\clubsuit_{\BB}(\WW'))$ into the first-in-first-out queue $\tilde{\WW}$\;

}

\end{algorithm}
Before describing Algorithm \ref{alg:framework-boundary},
we need to roughly define the term $\clubsuit_{\BB}(\WW_i)$ as a subset of raising fans in $\clubsuit_{\BB}$ which consist of all raising fans inside $\clubsuit_{\BB}(\WW_i)$ meet at a vertex on the right boundary of $\WW_i$.
Moreover, the subset of raising fans $\clubsuit(\WW_i)$ partition $\WW_i$ into  sequential regions and can be a $(\QQ, \XX, \DD)$-well-placed layout in $\HH$. (They are proven in Lemmas \ref{lem:partition-boundary} and \ref{lem:layout-boundary}, respectively.)
Initially, in Algorithm \ref{alg:framework-boundary}, the root $r$'s raising fan $\tilde{\FF}$ in a forest-like structure $\clubsuit_{\BB}$ form sequential
regions $\tilde{\WW}_{\HH}(\tilde{\FF})$ in $\HH$ and we put them into the first-in-first-out queue orderly.
Next each iteration $i$ of Algorithm \ref{alg:framework-boundary} picks up
the first region $\WW_i$ in the first-in-first-out queue $\tilde{\WW}$
(it means that the region $\WW_i$ at the leftmost region
such that the region $\WW_i$ consists of raising fans in $\clubsuit_{\BB}$ in $\HH$;
$|\clubsuit_{\BB}(\WW_i)|>0$)
to place the sequential regions $\tilde{\WW}(\clubsuit_{\BB}(\WW_i))$ (it can be proven in Lemma \ref{lem:partition-boundary}) at the rightmost side in $\HH$ as $\tilde{\WW}_{\HH}(\clubsuit_{\BB}(\WW_i))$. (It can be proven in Lemma \ref{lem:layout-boundary}.)
It means that all regions $\WW_j$ at left of $\WW_i$ in $\HH$ doesn't consist any raising fan in $\clubsuit_{\BB}$ ($|\clubsuit_{\BB}(\WW_j)|=0$) and the chords between $\WW_j$ and $\tilde{\WW}(\clubsuit_{\BB}(\WW_j))$ don't nest with
the chords between $\WW_i$ and
$\tilde{\WW}_{\HH}(\clubsuit_{\BB}(\WW_i))$
because $\WW_j$ is at left of $\WW_i$ in $\HH$ and
$\tilde{\WW}_{\HH}(\clubsuit_{\BB}(\WW_j))$ are at left of
$\tilde{\WW}_{\HH}(\clubsuit_{\BB}(\WW_i))$ in $\HH$.
It means that
the left end-vertices of the chords $\tilde{\WW}_{\HH}(\clubsuit_{\BB}(\WW_j))$ are placed at left of
the left end-vertices of the chords $\tilde{\WW}_{\HH}(\clubsuit_{\BB}(\WW_i))$, and
the right end-vertices of the chords $\tilde{\WW}_{\HH}(\clubsuit_{\BB}(\WW_j))$ are placed at left of
the right end-vertices of the chords $\tilde{\WW}_{\HH}(\clubsuit_{\BB}(\WW_i))$ in $\HH$.

Next, we give a precise definition of $\clubsuit_{\BB}(\WW_i)$ as follows:

\begin{definition}
Given a region $\WW_i=(\BB^L, \BB^R)$ and
a left forest-like structure $\clubsuit_{\BB}$,
$\clubsuit_{\BB}(\WW_i)$ is defined to consist of two sequential raising fans
$\clubsuit^L=(\tilde{\FF}^L_1, \tilde{\FF}^L_2, \cdots, \tilde{\FF}^L_p)$ and
$\clubsuit^R=(\tilde{\FF}^R_1, \tilde{\FF}^R_2, \cdots, \tilde{\FF}^R_q)$
where
\begin{itemize}

\item
$\clubsuit^L=(\tilde{\FF}^L_1, \tilde{\FF}^L_2, \cdots, \tilde{\FF}^L_p)$
are the maximal subsequential raising fans in $\clubsuit_{\BB}$
such that for each raising fan $\tilde{\FF}^L_i, 1\leq i\leq p$,
(1) $\tilde{\FF}^L_i$ does not touch any vertex on the right boundary $\BB^R$,
(2) $\tilde{\FF}^L_i$ is not a descendant of any raising fan $\tilde{\FF}^R_j, 1\leq j\leq q$, and
(3) $\tilde{\FF}^L_i$ is a right sibling of $\tilde{\FF}^L_{i-1}$ in $\clubsuit_{\BB}$. And,

\item
$\clubsuit^R=(\tilde{\FF}^R_1, \tilde{\FF}^R_2, \cdots, \tilde{\FF}^R_q)$
are the maximal subsequential raising fans in $\clubsuit_{\BB}$
such that for each raising fan $\tilde{\FF}^R_i, 1\leq i\leq q$,
$\tilde{\FF}^R_i$'s right wings touch the right boundary $\BB^R$.
\end{itemize}
\end{definition}

In the followings, we prove three properties: the first one proves that sequential raising fans $\clubsuit^R$ form a contiguous path in $\clubsuit_{\BB}$. Properties \ref{prop:left-contiguous} and \ref{prop:right-contiguous} state that for each vertex $v$, all rightward-outer fans for children of a vertex $v\notin \BB$ form a contiguous subsequence in $\clubsuit^R$.
Also, all rightward fans for a vertex $v\notin \BB$ also form contiguous subsequence in $\clubsuit^R$.

For any region $\WW_i\in \tilde{\WW}(\tilde{\FF^R_i})$,
$\WW_i$ has three different types: the type one is the $\WW_i$'s root is at a vertex on the left boundary of $\WW_{i-1}$. The type two is that $\BB^L(\WW_i)\cap \BB^L(\WW_{i-1})$ and $\BB^R(\WW_i)\cap \BB^R(\WW_{i-1})$ are sub-paths of $\BB^L(\WW_{i-1})$ and $\BB^R(\WW_{i-1})$ from the $\WW_{i-1}$'s root, respectively. The type three is that the $\WW_i$'s is at a vertex on the right boundary of $\WW_{i-1}$.
In the sequential regions $\tilde{\WW}(\tilde{\FF}^R_i)$, types one, two and three are orderly appeared from left to right.
Notes that the third type of region can bound contiguous raising fans of $\clubsuit^R$.

Given a region $\WW_i\in \tilde{\FF}^R_i$, the region $\WW_i$ is called a \emph{black-hole} if $\WW_i$ is the type two.
Intuitively, a black-hole is a region such that it can bound a contiguous subsequence $(\tilde{\FF}^R_{i+1}, \tilde{\FF}^R_{i+2}, \cdots, \tilde{\FF}^R_q)$ in $\clubsuit^R$.
The following observation states that for each raising fan in $\clubsuit^R$, there is the only one region which can be a black-hole.
Note that for a region $\WW_i$ in a raising fan $\tilde{\FF}^R_i$.
If the $\WW_i$'s root is at a vertex on the left boundary of $\WW_{i-1}$ or on the right boundary of $\WW_{i-1}$, then $\WW_i$ cannot consist of any raising fan in $\clubsuit^R$.
Hence the only region in $\tilde{\FF}_i$ which can consist of contiguous raising fans $(\tilde{\FF}^R_{i+1}, \tilde{\FF}^R_{i+2}, \cdots, \tilde{\FF}^R_q)$ in $\clubsuit^R$ is the $\tilde{\FF}$'s black-hole.
\begin{observation}\label{obs:black-hole}
If a region $\WW_i\in \tilde{\WW}(\tilde{\FF}^R_i)$ is a black-hole, then the region $\WW_i$ bounds subsequential raising fans $(\tilde{\FF}^R_{i+1}, \tilde{\FF}^R_{i+2}, \cdots, \tilde{\FF}^R_q)\subseteq \clubsuit^R$
where $(\tilde{\FF}^R_{i+1}, \tilde{\FF}^R_{i+2}, \cdots, \tilde{\FF}^R_q)\subseteq \clubsuit^R$
 form a contiguous subpath in $\clubsuit^R$ and $\clubsuit_{\BB}$.
And, $\WW_i$ is the only one black-hole in the raising fan $\tilde{\FF}^R_i$.
\end{observation}

\begin{property}
Sequential raising fans $\clubsuit^R=(\tilde{\FF}^R_1, \tilde{\FF}^R_2, \cdots, \tilde{\FF}^R_q)$
are an ancestor-descendant path in $\clubsuit_{\BB}$
such that
for each $1\leq i\leq q-1$,
a raising fan $\tilde{\FF}^R_i$ is the parent
of $\tilde{\FF}^R_{i+1}$ in $\clubsuit_{\BB}$.
($\tilde{\FF}^R_i$ bounds $\tilde{\FF}^R_{i+1}$.)
\end{property}

\begin{proof}
From Observation \ref{obs:black-hole},
we know that for each raising fan $\tilde{\FF}^R_i, 1\leq i\leq q$, there is the only one black-hole in the raising fan $\tilde{\FF}^R_i$.
Then
for each raising fan $\tilde{\FF}^R_i, 1\leq i\leq q-1$,
$\tilde{\FF}^R_i$ bounds the  raising fan $\tilde{\FF}^R_{i+1}$ in $\clubsuit^R$ and $\tilde{\FF}^R_i$ is the parent of $\tilde{\FF}^R_{i+1}$ in $\clubsuit_{\BB}$.
Hence we can prove that $\clubsuit^R=(\tilde{\FF}^R_1, \tilde{\FF}^R_2, \cdots, \tilde{\FF}^R_q)
$ is a path in $\clubsuit_{\BB}$.
\end{proof}

\begin{property}\label{prop:left-contiguous}
Sequential raising fans $\clubsuit^R$ have the following property:
for each vertex $v \notin \BB$ in the right wing of a raising fan $\tilde{\FF} \in \clubsuit^R$,
\begin{enumerate}
\item all rightward fans for the vertex $v$ form at most one raising fan in $\clubsuit^R$ and

\item all rightward fans for $v$'s children form at most one raising fan in $\clubsuit^R$.
\end{enumerate}
\end{property}

\begin{proof}
Let $v$ be a vertex at the right boundary of $\tilde{\FF}$'s right wing.
All rightward fans for the vertex $v$ and rightward-outer fans for $v$'s children are consisted in at most one raising fan $\tilde{\FF}' \in \clubsuit^R$.
We know that no any raising fan in $\clubsuit^R$ can be inside a region formed by $\tilde{\FF}'$'s left wing and it leads that the right boundary of $\tilde{\FF}$'s left wing cannot overlap $\tilde{\FF}'$'s right wing.

Hence we can conclude that for each vertex $v\notin \BB$ that $v$ is a vertex of the right wing of $\tilde{\FF} \in \clubsuit^R$,
all rightward fans for the vertex $v$ and all rightward-outer fans for $v$'s children form at most one raising fan $\tilde{\FF}'$ in $\clubsuit^R$.
\end{proof}

\begin{property}\label{prop:right-contiguous}
Sequential raising fans $\clubsuit^R=(\tilde{\FF}^R_1, \tilde{\FF}^R_2, \cdots, \tilde{\FF}^R_q)$ have the following property:
for each vertex $v\notin \BB$ in the right wing of a raising fan $\tilde{\FF}^R_i\in \clubsuit^R$,
all $u$'s leftward-outer fans form contiguous raising fans $(\tilde{\FF}^R_i, \tilde{\FF}^R_{i+1}, \cdots, \tilde{\FF}^R_a) \in \clubsuit^R$.
\end{property}

\begin{proof}
Let $\WW'$ be a black-hole passes through a vertex $u$ where the vertex $u$ is on the right boundary $\BB^R(\WW')$ of $\WW'$.
Let $\tilde{\FF}^R_{a_1}$ be the first raising fan in $\clubsuit^R$ such that the black-hole $\WW_{a_1} \in \tilde{\WW}(\tilde{\FF}^R_{a_1})$ passes through the vertex $u$.
Observe that all $u$'s leftward-outer fans are shared by a maximal contiguous black-holes $(\WW_{a_1}, \WW_{a_1+2}, \cdots, \WW_{a_2})$ where
each black-hole $\WW_j, a_1\leq j\leq a_2$, passes through the vertex $u$ till the black-hole $\WW_{a_2+1}\in \tilde{\WW}(\tilde{\FF}^R_{a_2+1})$ doesn't pass through the vertex $u$.
Now we know that there are sequential raising fans $(\tilde{\FF}^R_{a_1}, \tilde{\FF}^R_{a_1+1}, \cdots, \tilde{\FF}^R_{a_2}) \subseteq \clubsuit^R$ such that
each raising fan $\tilde{\FF}^R_i, a_1\leq i\leq a_2,$ consists of a black-hole $\WW_i$.
When the black-hole in
$\tilde{\FF}^R_{a_2+1}$ doesn't passes through the vertex $u$,
each raising fan $\tilde{\FF}^R_j, a_2+1\leq j\leq q$, cannot consists of any $u$'s leftward-outer fan.
Hence we can prove that $u$'s leftward-outer fans are shared by contiguous raising fans $(\tilde{\FF}^R_{a_1}, \tilde{\FF}^R_{a_1+1}, \cdots, \tilde{\FF}^R_{a_2}) \subseteq \clubsuit^R$.
\end{proof}

Now we can construct sequential regions $\tilde{\WW}(\clubsuit_{\BB}(\WW))$ from the sequential raising fans $(\clubsuit^L, \clubsuit^R)$ in a region $\WW$ as follows:
firstly, because the sequential raising fans $(\tilde{\FF}^L_1, \tilde{\FF}^L_2, \cdots, \tilde{\FF}^L_p, \tilde{\FF}^R_1)$ are mutually disjoint (they don't have any ancestor-descendant relation in $\clubsuit_{\BB}$), the sequential raising fans $(\tilde{\FF}^L_1, \tilde{\FF}^L_2, \cdots, \tilde{\FF}^L_p, \tilde{\FF}^R_1)$ partition $\WW$ into sequential disjoint regions $(\tilde{\WW}(\tilde{\FF}^L_1),$ $ \tilde{\WW}(\tilde{\FF}^L_2),$ $\cdots,$ $\tilde{\WW}(\tilde{\FF}^L_p),$ $\tilde{\WW}(\tilde{\FF}^R_1))$.
Secondly, for each raising fan $\tilde{\FF}^R_i, 2\leq i\leq q$, process the following steps:
let $\WW_i$ be a region in $\tilde{\WW}(\tilde{\FF}^R_i)$ bounds a raising fan $\tilde{\FF}^R_{i+1}$.
Then replace $\WW_i$ by $\tilde{\WW}(\tilde{\FF}^R_{i+1})$.
After the last raising fan $\tilde{\FF}^R_q$ is processed,
we get sequential disjoint regions $\tilde{\WW}(\clubsuit_{\BB}(\WW))$.

Now we show how to place sequential regions $\tilde{\WW}(\clubsuit_{\BB}(\WW))$ in $\HH$ as $\tilde{\WW}_{\HH}(\clubsuit_{\BB}(\WW))$ and
utilize it to place sequential regions $\tilde{\WW}(\clubsuit_{\BB})$ in $\HH$ as $\tilde{\WW}_{\HH}(\clubsuit_{\BB})$
such that $\tilde{\WW}_{\HH}(\clubsuit_{\BB})$ are $(\QQ, \XX, \DD)$-well-placed in $\HH$.

\begin{enumerate}
\item Firstly, we place the sequential regions $(\tilde{\WW}(\FF^L_1), \tilde{\WW}(\FF^L_2), \cdots, \tilde{\WW}(\FF^L_p), \tilde{\WW}(\FF^R_1))$ as $(\tilde{\WW}_{\HH}(\FF^L_1),$ $\tilde{\WW}_{\HH}(\FF^L_2),$ $\cdots,$
$\tilde{\WW}_{\HH}(\FF^L_p),$ $\tilde{\WW}_{\HH}(\FF^R_1))$ in $\HH$.
\item Because each raising fan $\tilde{\FF}^R_i, 2\leq i\leq q$, is inside a region $\WW_i$ of $\tilde{\WW}_{\HH}(\tilde{\FF}^R_{i-1})$, we place the sequential regions $\tilde{\WW}_{\HH}(\tilde{\FF}^R_i)$ at right of the sequential regions $\tilde{\WW}_{\HH}(\tilde{\FF}^R_{i-1})$ in $\HH$.

\end{enumerate}
Now we have new sequential regions in $\HH$ as follows: $\tilde{\WW}_{\HH}(\clubsuit_{\BB}(\WW))=$
$(\tilde{\WW}_{\HH}(\tilde{\FF}^L_1),$ $\tilde{\WW}_{\HH}(\tilde{\FF}^L_2),$
$\cdots,$
$\tilde{\WW}_{\HH}(\tilde{\FF}^L_p),$ $\tilde{\WW}_{\HH}(\tilde{\FF}^R_1),$ $\tilde{\WW}_{\HH}(\tilde{\FF}^R_2),$
$\cdots,$
$\tilde{\WW}_{\HH}(\tilde{\FF}^R_q))$.
And, we immediately have the following lemma.

\begin{lemma}\label{lem:partition-boundary}
Regions $\tilde{\WW}_{\HH}(\clubsuit_{\BB}(\WW))$
are sequential in $\HH$.
\end{lemma}

The next lemma proves that
the sequential regions $\tilde{\WW}_{\HH}(\clubsuit_{\BB}(\WW))=$
$(\tilde{\WW}_{\HH}(\tilde{\FF}^L_1),$ $\tilde{\WW}_{\HH}(\tilde{\FF}^L_2),$
$\cdots,$
$\tilde{\WW}_{\HH}(\tilde{\FF}^L_p),$ $\tilde{\WW}_{\HH}(\tilde{\FF}^R_1),$ $\tilde{\WW}_{\HH}(\tilde{\FF}^R_2),$
$\cdots,$
$\tilde{\WW}_{\HH}(\tilde{\FF}^R_q))$
are $(\QQ, \XX, \DD)$-well-placed in $\HH$.

\begin{lemma}\label{lem:layout-boundary}
Sequential regions $\tilde{\WW}_{\HH}(\clubsuit_{\BB}(\WW))$ are $(\QQ, \XX, \DD)$-well-placed in $\HH$.
\end{lemma}

\begin{proof}

Recall that $\clubsuit^R=(\tilde{\FF}^R_1, \tilde{\FF}^R_2, \cdots, \tilde{\FF}^R_q)$ are sequentially placed in $\HH$
as $(\tilde{\WW}_{\HH}(\tilde{\FF}^R_1)$, $\tilde{\WW}_{\HH}(\tilde{\FF}^R_2),$ $\cdots,$ $\tilde{\WW}_{\HH}(\tilde{\FF}^R_q))$
such that each raising fan $\tilde{\FF}^R_i, 2\leq i\leq q$, is at right of $\tilde{\FF}^R_{i-1}$ in $\HH$.

In the followings, we utilize
Property \ref{prop:left-contiguous} to prove that
for each vertex $v\notin \BB$,
edges between $v$ and $v$'s children have constant
$X$-crossing edges with other edges in $(\tilde{\WW}_{\HH}(\tilde{\FF}^R_1)$, $\tilde{\WW}_{\HH}(\tilde{\FF}^R_2),$ $\cdots,$ $\tilde{\WW}_{\HH}(\tilde{\FF}^R_q))$.
Because all $v$' rightward-outer fans form contiguous raising fans $(\tilde{\FF}_i, \tilde{\FF}_{i+1})$
with length at most two in $(\tilde{\FF}^R_1, \tilde{\FF}^R_2, \cdots, \tilde{\FF}^R_q)$.
The edges between $v$ and $v$'s children make $X$-crossing edges with
the only raising fan $\tilde{\WW}_{\HH}(\tilde{\FF}_{i+1})$ in $\HH$.
Hence the edges between $v$ and $v$'s children have $X$-crossing number at most one.

As we place $\clubsuit^R=(\tilde{\FF}^R_1, \tilde{\FF}^R_2, \cdots, \tilde{\FF}^R_q)$ as
$(\tilde{\WW}_{\HH}(\tilde{\FF}^R_1),$
$\tilde{\WW}_{\HH}(\tilde{\FF}^R_2),$
$\cdots,$
$\tilde{\WW}_{\HH}(\tilde{\FF}^R_q))$
in $\HH$,
we have sequential edges
$(\tilde{e}_1, \tilde{e}_2, \cdots, \tilde{e}_{q-1})$ that
each edge $e\in \tilde{e}_i, 1\leq i\leq q-1$, connects between two raising fans $\tilde{\FF}^R_i$ and
$\tilde{\FF}^R_{i+1}$ except for
edges between $v$ and $v$'s children,
Observe that
$(\tilde{e}_1, \tilde{e}_2, \cdots, \tilde{e}_{q-1})$
are orderly placed in $\HH$.
So, there is no any $X$-crossing edge among
$(\tilde{e}_1, \tilde{e}_2, \cdots, \tilde{e}_{q-1})$.

Since any two raising fans in $(\tilde{\FF}^L_1,$ $\tilde{\FF}^L_2,$
$\cdots,$
$\tilde{\FF}^L_p)$ are siblings in $\clubsuit_{\BB}$ (two raising fans $\tilde{\FF}$ and $\tilde{\FF}'$ are siblings in $\clubsuit_{\BB}$, $\tilde{\FF}$ and $\tilde{\FF}'$ are not bounded to each other),
the placement: $(\tilde{\WW}_{\HH}(\tilde{\FF}^L_1),$
$\tilde{\WW}_{\HH}(\tilde{\FF}^L_2),$ $\cdots,$
$\tilde{\WW}_{\HH}(\tilde{\FF}^L_p))$
are $(\QQ, \XX, \DD)$-well-placed in $\HH$.
Also,
because $(\tilde{\FF}^L_1, \tilde{\FF}^L_2, \cdots, \tilde{\FF}^L_p)$ and $(\tilde{\FF}^R_1, \tilde{\FF}^R_2, \cdots, \tilde{\FF}^R_q)$ are not bounded to each other,
except for the edges between a
vertex on a black-hole,
the placement: $(\tilde{\WW}_{\HH}(\tilde{\FF}^L_1), \tilde{\WW}_{\HH}(\tilde{\FF}^L_2), \cdots, \tilde{\WW}_{\HH}(\tilde{\FF}^L_p))$
and
$(\tilde{\WW}_{\HH}(\tilde{\FF}^R_1),$ $\tilde{\WW}_{\HH}(\tilde{\FF}^R_2),$ $\cdots,$
$\tilde{\WW}_{\HH}(\tilde{\FF}^R_q))$
are $(\QQ, \XX, \DD)$-well-placed in $\HH$.

In the followings, we utilize Property \ref{prop:right-contiguous} to prove that
for each vertex $v\notin \BB$,
all $v$'s leftward-outer fans have constant number of
$X$-crossing edges in $\HH$.
Let $v$ be a vertex on the right boundary of a black-hole.
By Property \ref{prop:right-contiguous},
there exists contiguous raising fans $(\tilde{\FF}^R_i, \tilde{\FF}^R_{i+1}, \cdots, \tilde{\FF}^R_a) \subseteq \clubsuit^R$ which consists of $v$'s leftward-outer fans.
All edges between $v$ and $v$'s children cross from the sequential raising fans
$(\tilde{\WW}_{\HH}(\tilde{\FF}^R_i),$ $\tilde{\WW}_{\HH}(\tilde{\FF}^R_{i+1}),$
$\cdots,$
$\tilde{\WW}_{\HH}(\tilde{\FF}^R_a))$ in $\HH$.

Let the sequential vertices $(v_1, v_2, \cdots, v_h)$ be orderly on a track in $\HH$
such that
each vertex $v_j, j\geq 1,$ consists of some leftward-outer fans.
Let $(\WW_1, \WW_2, \cdots, \WW_h)$ be sequential black-holes in $\HH$ such that
each vertex $v_j, 1\leq j\leq h$, is on the right boundary of $\WW_j$.

We know that
for each vertex $v_j, 1\leq j\leq h$, there are contiguous raising fans
$(\tilde{\FF}_{a_j},$ $\tilde{\FF}_{a_j+1},$
$\cdots$
$\tilde{\FF}_{b_j})$
passing through $v_j$.
Also, since the sequential black-holes $(\WW_1, \WW_2, \cdots, \WW_h)$ are disjoint, we have the following ordered relation
$(a_1 < a_2 \cdots < a_h)$.
Now,
the sequential vertices $(v_1, v_2, \cdots, v_h)$ are orderly placed on a track in $\HH$ and
the sequential children $(\CC_{v_1}, \CC_{v_2}, \cdots, \CC_{v_h})$ of $(v_1, v_2, \cdots, v_h)$
are also orderly placed on other track in $\HH$
because the order relation
$(a_1\leq b_1 < a_2 \leq b_2 \cdots < a_h \leq b_h)$.

Hence for all sequential edges
$(\tilde{e}(v_1),$
$\tilde{e}(v_2),$
$\cdots,$
$\tilde{e}(v_h))$
where each edges $\tilde{e}(v_j), 1\leq j\leq h$, are edges between
the vertex $v_j$ and and $v_j$'s children $\CC_{v_j}$,
$(\tilde{e}(v_1),$
$\tilde{e}(v_2),$
$\cdots,$
$\tilde{e}(v_h))$ are not $X$-crossing in $\HH$.

Finally, when we add the sequential edges
$(\tilde{e}(v_1), \tilde{e}(v_2), \cdots, \tilde{e}(v_h))$ into $\HH$,
the $X$-crossing number in $\HH$ increase one in $\HH$
because the sequential edges $(\tilde{e}(v_1), \tilde{e}(v_2), \cdots, \tilde{e}(v_h))$
make $X$-crossing edges with the sequential regions $(\WW_1, \WW_2, \cdots, \WW_h)$ in $\HH$.
and
the layout:
$(\tilde{\WW}_{\HH}(\tilde{\FF}^L_1),$ $\tilde{\WW}_{\HH}(\tilde{\FF}^L_2),$
$cdots,$
$\tilde{\WW}_{\HH}(\tilde{\FF}^L_p),$
$\tilde{\WW}_{\HH}(\tilde{\FF}^R_1),$ $\tilde{\WW}_{\HH}(\tilde{\FF}^R_2),$
$\cdots,$
$\tilde{\WW}_{\HH}(\tilde{\FF}^R_q))$
are $(\QQ, \XX, \DD)$-well-placed in $\HH$.
\end{proof}

\begin{lemma}\label{lem:full-layout-boundary}
Sequential regions $\tilde{\WW}_{\HH}(\clubsuit_{\BB})$ are $(\QQ, \XX, \DD)$-well-placed in $\HH$.
\end{lemma}

\begin{proof}
Initially, the root $r$'s raising fan $\tilde{\FF}$ in a forest-like structure $\clubsuit_{\BB}$ form sequential
regions $\tilde{\WW}_{\HH}(\tilde{\FF})$ in $\HH$.
Next each iteration of Algorithm \ref{alg:framework-boundary} picks up
the first region $\WW_i$ in the first-in-first-out queue $\tilde{\WW}$
to place the sequential regions $\tilde{\WW}(\clubsuit_{\BB}(\WW_i))$ at the rightmost side in $\HH$; The region $\WW_i$ at the leftmost region consists of raising fans in $\clubsuit_{\BB}$ in $\HH$.
($|\clubsuit_{\BB}(\WW_i)|>0$.)

All regions $\WW_j$ at left of $\WW_i$ in $\HH$ doesn't consist of any raising fan in $\clubsuit_{\BB}$ ($|\clubsuit_{\BB}(\WW_j)|=0$) and the chords between $\WW_j$ and $\tilde{\WW}(\clubsuit_{\BB}(\WW_j))$ don't nest with
the chords between $\WW_i$ and
$\tilde{\WW}_{\HH}(\clubsuit_{\BB}(\WW_i))$
because $\WW_j$ is at left of $\WW_i$ in $\HH$ and
$\tilde{\WW}_{\HH}(\clubsuit_{\BB}(\WW_j))$ are at left of
$\tilde{\WW}_{\HH}(\clubsuit_{\BB}(\WW_i))$ in $\HH$.
The above description also implies that
the left end-vertices of the chords $\tilde{\WW}_{\HH}(\clubsuit_{\BB}(\WW_j))$ are placed at left of
the left end-vertices of the chords $\tilde{\WW}_{\HH}(\clubsuit_{\BB}(\WW_i))$, and
the right end-vertices of the chords $\tilde{\WW}_{\HH}(\clubsuit_{\BB}(\WW_j))$ are placed at left of
the right end-vertices of the chords $\tilde{\WW}_{\HH}(\clubsuit_{\BB}(\WW_i))$ in $\HH$.

Observe that if sequential regions
$(\tilde{\WW}(\clubsuit_{\BB}(\WW_1)),$
$\tilde{\WW}(\clubsuit_{\BB}(\WW_2)),$
$\cdots)$ are orderly placed in
$\HH$ as
$(\tilde{\WW}_{\HH}(\clubsuit_{\BB}(\WW_1)),$
$\tilde{\WW}_{\HH}(\clubsuit_{\BB}(\WW_2)),$
$\cdots)$ in Algorithm \ref{alg:framework-boundary}, then
the sequential regions $(\WW_1, \WW_2, \cdots)$ are also orderly
placed in $\HH$.
Let $(\WW_1, \WW_2, \cdots)$
be sequential regions in $\HH$
and $(\tilde{e}_1, \tilde{e}_2, \cdots)$ be
sequential chords in $\HH$
where $\tilde{e}_i, i\geq 1,$ are edges between $\WW_i$ and $\clubsuit_{\BB}(\WW_i))$.
When we place a new sequential regions
$\tilde{\WW}_{\HH}(\clubsuit_{\BB}(\WW_i))$ at right of
the sequential regions
$(\tilde{\WW}_{\HH}(\clubsuit_{\BB}(\WW_1)),$
$\tilde{\WW}_{\HH}(\clubsuit_{\BB}(\WW_2)),$
$\cdots,$
$\tilde{\WW}_{\HH}(\clubsuit_{\BB}(\WW_{i-1})))$,
the sequential chords $(\tilde{e}_1,$
$\tilde{e}_2,$
$\cdots,$
$\tilde{e}_i)$
don't nest to each other in $\HH$
because
the order of $(\WW_1,$
$\WW_2,$
$\cdots,$
$\WW_i)$ in $\HH$ is the same as the order of
$(\tilde{\WW}_{\HH}(\clubsuit_{\BB}(\WW_1)),$
$\tilde{\WW}_{\HH}(\clubsuit_{\BB}(\WW_2)),$
$\cdots,$
$\tilde{\WW}_{\HH}(\clubsuit_{\BB}(\WW_i)))$ in $\HH$.
Hence the layout in Algorithm \ref{alg:framework-boundary}
is $(\QQ, \XX, \DD)$-well-placed in $\HH$.
\end{proof}

In Algorithm \ref{alg:framework-boundary}, we place a forest-like raising fans $\clubsuit_{\BB}$ of a boundary $\BB$ in a ladder $\HH$ where the input $\clubsuit_{\BB}$ consists of the only one root.
However, $\clubsuit_{\BB}$ would be a forest with sequential roots $(r_1, r_2, \cdots, r_s)$,
we can slightly modified Algorithm \ref{alg:framework-boundary} as follows:
if $\clubsuit_{\BB}$ is a forest with the sequential roots $(r_1, r_2, \cdots, r_s)$,
we can orderly place the sequential regions
$(\tilde{\WW}_{\HH}(r_1), \tilde{\WW}_{\HH}(r_2), \cdots, \tilde{\WW}_{\HH}(r_s))$ in $\HH$.

\begin{theorem}
Given a region $\WW=(\BB^L, \BB^R)$, the collection of all rightward-outer and leftward-outer fans from the boundaries $\BB^L$ and $\BB^R$
can be $(\QQ, \XX, \DD)$-well-placed in $\HH$ and
there are sequential regions $(\tilde{\WW}_{\HH}(\clubsuit_{\BB^L}), \WW^M, \tilde{\WW}_{\HH}(\clubsuit_{\BB^R}))$ in $\HH$
where $\WW^M$ is a region between
the rightmost boundary of $\tilde{\WW}(\clubsuit_{\BB^L})$
and the leftmost boundary of
$\tilde{\WW}(\clubsuit_{\BB^R})$.
Moreover, the union of the two forest-like structures
$\clubsuit_{\BB^L}$ and $\clubsuit_{\BB^R}$ is a skeleton $\Psi(\WW)$ of $\WW$.
\end{theorem}

\begin{figure}[t]
\begin{center}
\includegraphics[width=1\textwidth, angle =0]{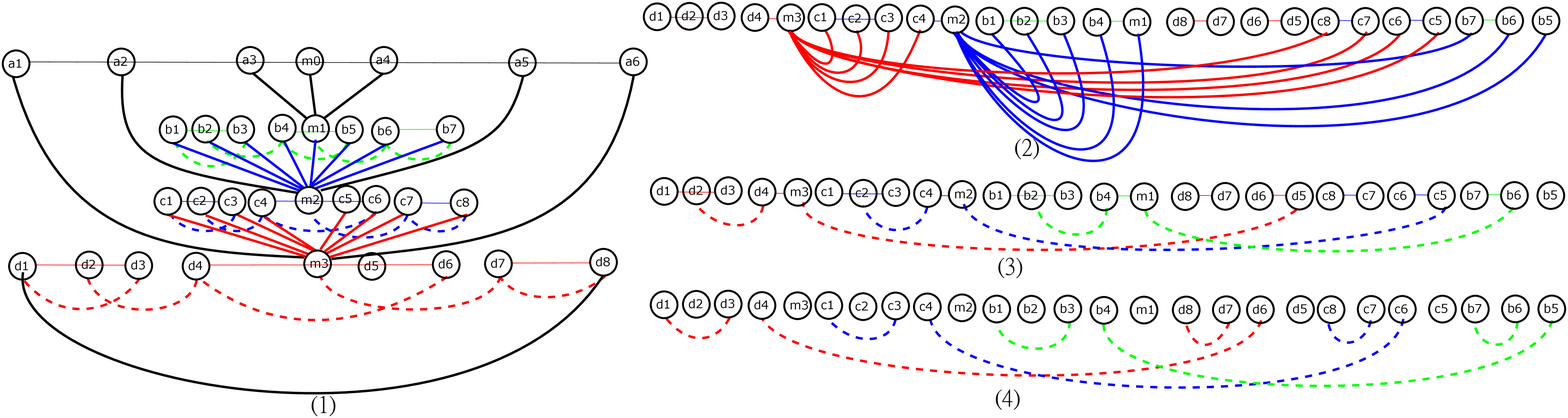}
  \centering
\caption{(1): $(\triangledown_1, \triangledown_2, \triangledown_3)$ are sequential down-pointing triangles that $\triangledown_1=(a_3, m_0, a_4, m_1)$,
$\triangledown_2=(a_2, a_6 m_2)$ and $\triangledown_3=(a_1, a_7, m_3)$.
The middle path $\MM=(m_0, m_1, m_2, m_3)$ partition $(\triangledown_1, \triangledown_2, \triangledown_3)$ into two disjoint parts.
By our algorithm, $(b_1, \cdots, b_4, m_1)$, $(c_1, \cdots, c_4, m_2)$ and $(d_1, \cdots, d_4, m_3)$ are placed on a track as the order: $(d_1, \cdots, d_4, m_3, c_1, \cdots, c_4, m_2,  b_1, \cdots, b_4, m_1)$.
The three sequential vertices $(b_5, \cdots, b_7)$, $(c_5, \cdots, c_8)$ and $(d_5, \cdots, d_8)$ are placed on a track as the order: $(d_8, \cdots, d_5, c_8, \cdots, c_5, b_7, \cdots, d_5)$;
(2): $m_3$'s left piles are
$\{(m_3, c_1), (m_3, c_2), (m_3, c_3), (m_3, c_4)\}$ and
$m_2$'s left piles are $\{(m_2, b_1), (m_2, b_2), (m_2, b_3),$ $(m_2, b_4)\}$. Because $(m_3, m_2, m_1)$ are orderly placed in $\HH$ and the sequential vertices $(c_1, c_2, c_3, c_3, c_4)$ are placed at left of the sequential vertices $(b_1, b_2, b_3, b_4)$, $m_2$'s left piles cannot make nested with $m_3$'s left piles in $\HH$.
Similarly,
$m_3$'s right piles are $\{(m_3, c_5), (m_3, c_6), (m_3, c_7), (m_3, c_8)\}$, and
$m_2$'s right piles are $\{(m_2, b_5), (m_2, b_6), (m_2, b_7)\}$.
Because the sequential vertices $(c_5, c_6, c_7, c_8)$ are placed at left of the sequential vertices $(b_5, b_6, b_7)$ in $\HH$,
$m_2$'s right piles cannot make nested with $m_3$'s right piles in $\HH$;
(3): The bridges $((d_2, d_4), (c_2, c_4), (b_2, b_4))$ at left of the middle path $\MM$ are orderly placed on a track in $\HH$. Hence
$((d_2, d_4), (c_2, c_4), (b_2, b_4))$ cannot form nested chords on a track in $\HH$.
The bridges $((m_3, d_7), (m_2, c_7), (m_1, b_6))$ that are connecting to middle path $\MM$ are also orderly placed, so $((m_3, d_7), (m_2, c_7), (m_1, b_6))$ cannot form nested chords on a track in $\HH$;
(4): The chord $((d_1, d_3), (c_1, c_3), (b_1, b_3))$ at left of the middle path $\MM$ are orderly placed on a track in $\HH$.
The right chord $((d_7, d_8), (c_7, c_8), (b_6, b_7))$ at right of middle path are also orderly placed on a track in $\HH$.
The chord $((d_4, d_6), (c_4, c_6), (b_4, b_5))$ across the middle path $\MM$ are not nested on a track in $\HH$ because both of the sequential vertices $(d_4, c_4, b_4)$ and
$(d_6, c_6, b_5)$ are orderly placed on a track in $\HH$.
}
\label{fig:ugly-edges-layout}
\end{center}
\vspace{-0.2in}
\end{figure}

\section{Deleted Edges Increase $X$-Crossing, Queue and Gap Numbers Sightly}

In this section, we explain why our layout is still $(\QQ, \XX, \DD)$-well-placed after deleted edges are re-added into our layout.


Recall that in the Section \ref{sec:transformation}, for a cycle $\OO$,
the cycle $\OO$ is clockwisely placed from the vertex $m$ of $\OO$ in a composite-layerlike graph $\GG$.
And,
for a down-pointing triangle $\triangledown$,
the down-pointing triagle $\triangledown$ is clockwisely placed from the lower vertex $m$ of $\triangledown$ where
all vertices except the lower vertex $m$ of $\triangledown$ are placed at a upper layer
and the lower vertex $m$ is placed at a lower layer of a composite-layerlike graph $\GG$.
Also, there exist sequential maximal inner cycles $(\OO_1, \OO_2, \cdots, \OO_p)$ inside $\OO$ or $\triangledown$;
The spines of the cycle $\OO$ or
the down-pointing triangle $\triangledown$.
And, the leftmost and rightmost vertices of the $i$-th cycle $\OO_i \in (\OO_1, \OO_2, \cdots, \OO_p)$ are the $i$-th joint of the spine.
Moreover,
for the $i$-th cycle $\OO_i \in (\OO_1, \OO_2, \cdots, \OO_p)$,
the $i$-th \emph{hoop}
$(u_i, u'_i)$ on the cycle $\OO$ is defined that the vertices $u_i$ and $u'_i$ are the parents of the $i$-th joint of the spine of the cycle $\OO$ or the down-pointing triangle $\triangledown$.
The vertex $m$ is called the \emph{bad} vertex of a cycle $\OO$
or a down-pointing triangle $\triangledown$.

\begin{definition}
For a cycle $\OO$,
a deleted edge on the bad vertex $m$ of the cycle $\OO$ are called a \emph{wire} of $m$ inside $\OO$ denoted as $\curlywedge(m)$. Also, a removed edge $e_i, 1\leq i\leq p-1,$ which connects between two contiguous maximal inner cycles $\OO_i$ and $\OO_{i+1}$ of the spine is called a \emph{bridge}.
\end{definition}

Recall that in Algorithm \ref{alg:framework},
we simultaneously
pick all joints of $\OO_i, 1\leq i\leq p$ and their hoops 
$\{(u_1, u'_1), (u_2, u'_2), \cdots, (u_p, u'_p)\}$
on the two contiguous tracks in $\HH$.
Also, we order all hoops $\{(u_1, u'_1),$
$(u_2, u'_2),$
$\cdots,$
$(u_p, u'_p)\}$ as ordering the lower boundary $\{L^B(\OO_1), L^B(\OO_2), \cdots, L^B(\OO_p)\}$ of the $\OO$'s spine.
Moreover, for all hoops $((u_1, u'_1),$
$(u_2, u'_2),$
$\cdots,$
$(u_p, u'_p))$ and the lower boundary
$\{L^B(\OO_1),$
$L^B(\OO_2),$
$\cdots,$
$L^B(\OO_p)\}$ of the spine, we place the two sequential vertices contiguously on any track in $\HH$.
Hence we have the following observations that state the key reasons why wires cannot make $X$-crossing in our layout.

\begin{observation}
Given a bad vertex $m$ on a cycle $\OO$,
let $\{\OO_1, \OO_2, \cdots, \OO_p\}$ be $\OO$'s spine and $\{(u_1, u'_1),$
$(u_2, u'_2),$
$\cdots,$
$(u_p, u'_p)\}$
be corresponding hoops on $\OO$, the layout in $\HH$ has the following properties:
\begin{enumerate}
\item all hoops are placed contiguously on the same track in $\HH$,

\item all joints are placed contiguously on a track in $\HH$,

\item the order of all joints on a track is the same as the order of all hoops on a track in $\HH$. And,

\item all hoops and all joints are placed at two contiguous tracks in $\HH$.
\end{enumerate}
\end{observation}

\begin{observation}
Let $m$ and $m'$ be bad vertices on cycles $\OO$
and $\OO'$, respectively.
The bad vertex $m$ is placed at left of the bad vertex $m'$ on a track in $\HH$ if and only if
the spine of $\OO$ is placed at
left of the spine of $\OO'$ at any track in $\HH$.
\end{observation}

\begin{observation}
For a cycle $\OO$ with the bad vertex $m$,
\begin{enumerate}
\item the gap number between the bad vertex $m$ and any vertex on the lower boundary $\{L^B(\OO_1),$
$L^B(\OO_2),$
$\cdots,$
$L^B(\OO_p)\}$ of the $\OO$'s spine is at most $2\JJ$, and

\item each lower boundary $L^B(\OO_i),$
$1\leq i\leq p$
except its joint
is placed contiguously on a track in $\HH$.
\end{enumerate}
\end{observation}

For any two bad vertices $m_1$ and $m_2$ that $m_1$ is at left of $m_2$ on a track in $\HH$,
the wires $\curlywedge(m_1)$ are placed at left of the wires $\curlywedge(m_2)$ in $\HH$.
So, there is no any $X$-crossing edge between $\curlywedge(m_1)$ and $\curlywedge(m_2)$.
From the above observations,
we have the following lemma:

\begin{lemma}\label{lem:vertical-X-crossing}
Suppose sequential bad vertices $(m_1, m_2, \cdots, m_p)$ are placed from left to right on a track in $\HH$,
the sequential wires $(\curlywedge(m_1),$ $\curlywedge(m_2),$ $\cdots,$ $\curlywedge(m_p))$ are not $X$-crossing in $\HH$.
\end{lemma}

\begin{definition}
Suppose (1) $\triangledown$ is a down-pointing triangle  with the bad vertex $m$ and
(2) the down-pointing triangle $\triangledown$ is in a fan $\FF$ of a raising fan $\tilde{\FF}$ with the middle path $\MM$,
\begin{enumerate}
\item a left (right, respectively) \emph{pile} of the bad vertex $m$
is defined as an edge connecting between the bad vertex $m$ and a vertex in the lower boundary $\{L^B(\OO_1), L^B(\OO_2), \cdots, L^B(\OO_p)\}$ of the spine of $\triangledown$ that
is at left (right, respectively) of the middle path $\MM$.
A left (right, respectively) pile with respect to the middle path $\MM$ is denoted to $\curlyvee^L(m)$ ($\curlyvee^R(m)$, respectively).

\item
A left (right, respectively) spine of the down-pointing triangle $\triangledown$ with respect to the middle path $\MM$ is the subsequential spine of the down-pointing triangle $\triangledown$ at left (right, respectively) of the middle path $\MM$.

\item
A left (right, respectively) hoop with respect to the middle path $\MM$ is a hoop of the down-pointing triangle $\triangledown$ at left (right, respectively)
of the middle path $\MM$.
\end{enumerate}
\end{definition}

Suppose $\triangledown'$ is a down-pointing triangle with the bad vertex $m'$
inside the down-pointing triangle $\triangledown$.
Then
the sequential regions $\tilde{\WW}(m)$ consisting of all $m$'s left joints are placed at left of the sequential regions $\tilde{\WW}(m')$ consisting of all $m'$'s left joints.
Because we place the bad vertex $m$ at left of the bad vertex $m'$ on a track in $\HH$ and
all $m$'s joints at left of $m'$'s hoops on a track in $\HH$,
we can have that the left piles $\curlyvee^L(m)$ are not nested with the left piles $\curlyvee^L(m')$ on any track in $\HH$.
Similarly,
the sequential regions $\tilde{\WW}(m)$ consisting of all $m$'s right joints are at left of the sequential regions $\tilde{\WW}(m')$ consisting of all $m'$'s right joints.
Because we place the bad vertex $m$ at left of
bad vertex $m'$ on a track in $\HH$ and
place
all $m$'s right joints at left of
all $m'$'s right joints on a track in $\HH$,
there is no any nested edge between
all right piles $\curlyvee^R(m)$ of the bad vertex $m$ and all right piles $\curlyvee^R(m')$ of the bad vertex $m'$.
Now we can the following observations:
\begin{observation}
Suppose raising down-pointing triangles $(\triangledown_1, \triangledown_2, \cdots, \triangledown_p)$ and their sequential bad vertices
$\MM=(m_1, m_2, \cdots, m_p)$ are placed on the same track in $\HH$,
\begin{enumerate}
\item the sequential left joints of the sequential down-pointing triangles $(\triangledown_1, \triangledown_2, \cdots, \triangledown_p)$ are orderly placed at a track in $\HH$. And,

\item the sequential right joints of the sequential down-pointing triangles $(\triangledown_1, \triangledown_2, \cdots, \triangledown_p)$ are orderly placed at a track in $\HH$.
\end{enumerate}
\end{observation}

\begin{observation}
For a down-pointing triangle $\triangledown$ with the bad vertex $m$,
\begin{enumerate}
\item the gap number between the bad vertex $m$ and any vertex on the lower boundary
$\{L^B(\OO_1),$
$L^B(\OO_2),$
$\cdots,$
$L^B(\OO_p)\}$ of the $\triangledown$'s spine is at most $2\JJ$, and

\item each lower boundary $L^B(\OO_i), 1\leq i\leq p$
except its joint is placed contiguously on a track in $\HH$.
\end{enumerate}
\end{observation}

From the above observations, we can have that the right piles $\curlyvee^R(m_i)$ are not nested with the left piles $\curlyvee^R(m_j)$ on any track in $\HH$ in the following lemma:

\begin{lemma}\label{lem:pile-nested}
Given sequential bad vertices $(m_1, m_2, \cdots, m_p)$ orderly placed on a track in $\HH$,
their sequential left and right piles $(\curlyvee^L(m_1), \curlyvee^L(m_2), \cdots, \curlyvee^L(m_p))$ and $(\curlyvee^R(m_1), \curlyvee^R(m_2), \cdots, \curlyvee^R(m_p))$ are not nested on any track in $\HH$.
\end{lemma}

For a cycle $\OO$,
all joints of the spine of the cycle are orderly placed on any track
in $\HH$,
all bridges of the spine cannot have nested chords on any track in $\HH$.
Similarly,
for a down-pointing triangle $\triangledown$,
all left and right joints of the left and right spines of the cycle are orderly placed on any track in $\HH$, respectively,
all bridges of the spine cannot have nested chords on any track in $\HH$.
From the above fact, we can have the following lemma:
\begin{lemma}\label{lem:bridge-nested}
Given a cycle $\OO$ or a down-pointing triangle $\triangledown$ with their  spine $(\OO_1, \OO_2, \cdots, \OO_p)$,
their sequential bridges $(e_1, e_2, \cdots, e_{p-1})$ are not nested on any track in $\HH$
where $e_i, 1\leq i\leq p-1,$ is the bridge between
the cycles $\OO_i$ and $\OO_{i+1}$.
\end{lemma}

\begin{theorem}\label{thm:G1-well-placed}
Every $1$-subdivision plane graph $G^1$ can have an $(\QQ, \XX, \DD)$-well-placed layout on constant number of tracks.
\end{theorem}

\begin{proof}
From Theorem \ref{thm:reform}, we know a plane graph $G$ can be reformed into a composite-layerlike graph $\GG$.
From Lemmas \ref{lem:vertical-X-crossing}, \ref{lem:pile-nested}
and \ref{lem:bridge-nested}, deleted edges slightly increase $X$-crossing number in any two tracks and queue number in any track in $\HH$.
Hence we conclude that
a plane graph $G$ can be $(\QQ, \XX, \DD)$-well-placed in a ladder $\HH$.
Also, from Theorem \ref{thm:cons-track},
an $(\QQ, \XX, \DD)$-well-placed layout can
be wrapped into a ladder $\HH$ on constant number of tracks.
\end{proof}

\bibliographystyle{plain}

\end{document}